\documentclass{article}

\usepackage[utf8]{inputenc} 
\usepackage[T1]{fontenc}    
\usepackage{hyperref}       
\usepackage{url}            
\usepackage{booktabs}       
\usepackage{amsfonts}       
\usepackage{nicefrac}       
\usepackage{microtype}      
\usepackage{xcolor}         
\usepackage[numbers]{natbib}
\usepackage{geometry}
\usepackage{float}
\usepackage{amsmath,amsfonts,amssymb,bbm,amsthm, bm,xspace,color,algorithm,algorithmic,graphicx}
\usepackage{thm-restate}
\usepackage[utf8]{inputenc}

\newtheorem{theorem}{Theorem}[section]
\newtheorem{lemma}[theorem]{Lemma}
\newtheorem{claim}[theorem]{Claim}
\newtheorem{corollary}[theorem]{Corollary}
\newtheorem{proposition}[theorem]{Proposition}

\theoremstyle{definition}
\newtheorem{definition}[theorem]{Definition}

\floatname{algorithm}{Mechanism}

\newcommand{\acc}{\textbf{A}\xspace}
\newcommand{\rej}{\textbf{R}\xspace}

\newcommand{\calR}{\mathcal{R}}
\newcommand{\calS}{\mathcal{S}}
\newcommand{\calW}{\mathcal{W}}
\newcommand{\calL}{\mathcal{L}}
\newcommand{\calH}{\mathcal{H}}

\DeclareMathOperator{\Bin}{Bin}
\DeclareMathOperator*{\E}{\mathbb{E}}

\newcommand{\gs}[1]{}
\newcommand{\bt}[1]{}
\newcommand{\full}[1]{}

\title{Wisdom of the Crowd Voting: Truthful Aggregation of Voter Information and Preferences}

\author{%
  Grant Schoenebeck\thanks{This material is based upon work supported by the National Science Foundation under Grant No. 2007256 }\\
  School of Information\\
  University of Michigan \\
  \texttt{schoeneb@umich.edu} \\
  \and
  Biaoshuai Tao \\
  John Hopcroft Center for Computer Science \\
  Shanghai Jiao Tong University \\
  \texttt{bstao@sjtu.edu.cn} \\
}
\date{}

\begin{document}
\maketitle

\begin{abstract}
    We consider two-alternative elections where voters' preferences depend on a state variable that is not directly observable.  
    Each voter receives a private signal that is correlated to the state variable.   Voters may be ``contingent'' with different preferences in different states; or predetermined with the same preference in every state.  In this setting, even if every voter is a contingent voter, agents voting according to their private information need not result in the adoption of the universally preferred alternative, because the signals can be systematically biased.
    
    We present an easy-to-deploy mechanism that elicits and aggregates the private signals from the voters, and outputs the alternative that is favored by the majority.  In particular, voters truthfully reporting their signals forms a strong Bayes Nash equilibrium (where no coalition of voters can deviate and receive a better outcome).
\end{abstract}

\section{Introduction}
\label{sect:intro}
Social choice theory studies how to aggregate participants' heterogeneous opinions/preferences and output a collective decision from a set of  alternatives.  
Typically, though not always, it is assumed that each participant has a clear preference over the alternatives, e.g., a preference order over all the alternatives, a valuation for each alternative, etc.
However, even with only two alternatives, this is not the typical case for \emph{all} participants. In addition to the participants who have clear, \emph{predetermined} preferences for one alternative over the other, typically, there are also \emph{contingent} participants who only have partial information on which alternative is ``preferable for them'' and yet would like to select the  alternative that is ``preferable for them.''

A standard example would be an election with two candidates $a$ and $b$ coming from political parties $A$ and $B$ respectively.  Voters are normally partitioned into three types: some partisans for Party $A$ prefer candidate $a$ based on his support for the platform of Party $A$; other partisans for Party $B$ prefer candidate $b$ based on her support for the platform of Party $B$;  finally, there are swing voters who are largely indifferent between the parties' platforms and would like to elect whichever candidate can make more progress on non-partisan issues.  However, swing voters do not have perfect information about which candidate is better suited for addressing the non-partisan needs of the community.  Instead, each voter has a hunch of which candidate will perform better on the non-partisan problems facing the community based on both public information and their private experiences and beliefs.

Additional examples where participants have preferences, but may or may not know what is ``preferable for them'', abound.  In votes for corporate strategies, for hiring decisions, and for policy decisions typically some participants would like to select according to some truth they are collectively trying to discern (e.g., impact on future profits, suitability for the position, efficacy of policy, etc.) while others may have predetermined preferences, for example, because of the way they are uniquely affected (e.g., prominence of their position in future corporate strategy, vision/skills of the job candidate, who in particular the policy benefits/harms).  

\subsection{Informal Setting}
In this paper, we consider a two-alternative social choice setting where voters' preferences may depend on a state variable.  In the above example, the state is which candidate will make more progress on non-partisan issues. We consider binary state variables in the main body of this paper.
In general, the state need not be binary.  For example, we could generalize the above example so which candidate would preform better on non-partisan issues is on a scale of 1 to 10 (where 1 indicates candidate $a$ is much better and 10 indicates candidate $b$ is much better).  Predetermined voters' preferences would still not depend on the state.  However, contingent voters' may have different thresholds on the state where they would transfer their support from candidate $a$ to candidate $b$.
In Appendix~\ref{append:nonbinary}, we discuss the non-binary setting and extend our results to this setting.

The state variable is not directly observable in the election phase, as the performance of a new government official, a new hired employee, a new policy, etc., may not be revealed until many years after the vote.
Instead, each voter receives a signal that is correlated to the hidden state which models the information voters received from difference sources.

Our goal is to select the \emph{majority wish}, the alternative that would be preferred by the majority if they knew the state of the world.  If some type of predetermined voters forms a majority, this is rather easy.  However, in the case where the predetermined voters of neither alternative forms a majority, the mechanism needs to aggregate the information and preferences of the contingent voters and selects the alternative which, for a majority of voters, is ``preferable for them''.

\subsection{Imperfectly Informed Voters}
\paragraph{Social Choice}
Social choice theory with imperfectly informed voters dates back to Condorcet's jury theorem in 1785~\citep{Condorcet1785}, and has also been widely studied~\citep{Miller1986,Young1988,ladha1992,nitzan2017collective}.
In Condorcet's setting, there are two alternatives, one of which is ``correct'', and each voter votes for the correct alternative with probability $p$.
Condorcet's jury theorem states that the probability that the majority voting scheme outputs the correct alternative goes to $1$ as the number of voters increases when $p>0.5$, and, conversely, this probability goes to $0$ when $p<0.5$.

Two unfortunate limitations for Condorcet's jury theorem are
1) It fails to output the correct alternative in the case voters' beliefs are aligned to the incorrect alternative (i.e., $p<0.5$), and
2) It assumes voters vote truthfully and disregards voters' potential strategic behaviors.  However, even in the case all voters have the same preference for the correct alternative, voting truthfully still may not be a Nash equilibrium~\citep{Austen1996}.

\gs{Here we talk about the "correct" alternative in the Feddersen paper.  Is is the correct alternative, or the majority wish?  Seems like we should be more careful.}  
\bt{\citet{Feddersen1994} consider the majority. However, I think these two notions can be used interchangeably in the context of Condorcet's jury theorem. We can view the ``correct'' alternative as the majority wish (i.e., all the agents prefer the correct alternative). It seems to me that this is exactly what \citet{Austen1996} is doing.}

To circumvent the first limitation, \citet{Feddersen1994} consider the scenario where voters play a Nash equilibrium strategy profile, while the voting rule is still the majority scheme.
\citet{Feddersen1994} show that when the number of voters tends to infinity, the probability that a majority voting scheme outputs the correct alternative approaches to $1$ if voters play the equilibrium strategy profile, while this probability is bounded away from $1$ if voters play the truthful strategy profile instead.
Feddersen and Pesendorfer's model assigns each voter a preference parameter $x\in[-1,1]$ describing his/her alignment to the two alternatives.
In their model, the unique (Bayes) Nash equilibrium is characterized by two thresholds $x_0,x_1\in[-1,1]$ with $x_0 < x_1$ such that voters with preferences below $x_0$ always vote for one alternative, voters with preferences above $x_1$ always vote for the other, and voters with preferences between $x_0$ and $x_1$ vote truthfully.
Although Feddersen and Pesendorfer's solution guarantees that the correct alternative is output with high probability, it requires sophisticated voters.
The voters need to calculate the values of $x_0$ and $x_1$ to decide their actions.  The values $x_0$ and $x_1$ are each the zero point of a continuous monotone function involving a complicated Riemann integral. This is usually too demanding for voters in practice, especially those who do not have a mathematical background.  Moreover, it needs to be common knowledge that all agents can and will perform this computation.  

In this paper, we take a different approach.
Instead of asking voters to play the Nash equilibrium for majority voting, we seek to design a more sophisticated voting scheme, or a mechanism, than the majority voting scheme, such that voters are incentivized to vote \emph{truthfully} under the mechanism, while guaranteeing the correct alternative is output with high probability.
Our social choice mechanism thus elicits truthful information from the voters and then aggregates it.

\paragraph{Information Aggregation}
The information aggregation literature considers how to obtain a ``correct answer'' by aggregating individuals' partial information---the crowd's wisdom.
The straightforward procedure of outputting the answer that is believed to be correct by the majority does not always work~\citep{chen2004eliminating,simmons2011intuitive}.  An example where this fails is when the crowd has a strong prior belief for the incorrect answer while novel specialized knowledge is only shared among a minority of the agents.  It is also known that further calibration based on collecting participants' confidences (the posterior of their beliefs) does not always solve this problem~\citep{hertwig2012tapping,prelec2017solution}.
In a seminal work by~\citet{prelec2017solution}, a new ``surprisingly popular'' approach was proposed: the participants' predictions over the other remaining participants' reported answers are collected, and the answer that is reported by more participants than predicted is output (we will review this in Section~\ref{sect:prelec_et_al}).  They justified this approach both theoretically and through experimentation.  In particular, they demonstrate the viability of approaches that require agents to predict reports of the other agents.
\citet{hosseinisurprisingly} empirically extend the surprisingly popular approach to the non-binary setting, where the goal is to learn the correct \emph{ranking over many options} instead of the correct answer in two options.

The work by \citet{prelec2017solution} and \citet{hosseinisurprisingly} does not fit into the social choice context in two aspects.
Firstly, the objective for an information aggregation mechanism is to output the correct answer.
Participants who collaboratively contribute their knowledge/information do not have preferences on which answer is finally selected.
This is fundamentally different from the social choice setting where the whole point of a social choice mechanism is to select an alternative favored by the majority.
Secondly, as agents care about the outcome, agents may be strategic and manipulate their reports in order to make their preferred alternatives win,
while \citet{prelec2017solution} and \citet{hosseinisurprisingly} do not put the problem in a game theory setting.

\subsection{Our Results}
In this paper, we study the social choice problem in a game theory setting with the existence of imperfectly informed voters who only have partial information regarding which alternative is more favorable.
For various settings with two alternatives, we propose a mechanism that aggregates participants' private information and outputs the alternative favored by more than half of the participants---\emph{the majority wish}.
Our mechanisms are truthful, in the sense that the truthful strategy profile forms a \emph{strong Bayes Nash Equilibrium}.

\full{Will have to switch the following paragraph back a bit for the full version.}
Our main wisdom-of-the-crowd-voting mechanism, presented in the paper, applies to the case of two worlds/states, where each agent receives a binary signal.  
We show it has strong truthfulness and aggregation properties, even for relatively small numbers of agents (Section~\ref{sect:twoworldstwosignals}).  
This result requires that the type distribution of agents is a common knowledge.   
In Sect.~\ref{sect:unknowndistribution}, we show that this common knowledge assumption is necessary to attain a strongly truthful mechanism that outputs the majority wish with high probability.
Our results for non-binary worlds and signals are deferred to Appendix~\ref{append:nonbinary}.
Specifically, we extend our results to the setting of more than two worlds in the limit as the number of agents grows in Appendix~\ref{sect:nonbinaryWorlds}.  
Finally, we show how to extend our results to the case of more than two signals in Appendix~\ref{sect:nonbinary-signals}.

Our mechanism can easily be implemented using a simple questionnaire that elicits voters' information and preferences.  The questions in the questionnaire are friendly to those voters who do not have relevant backgrounds in mathematics, game theory, etc., and rely on the same notions, require participants to predict the responses of other agents, which is empirically validated by the surprisingly popular method from \citet{prelec2017solution}.

We ensure our mechanisms have a group truthfulness property by employing a ``median trick.''  Intuitively, by the Median Voter Theorem~\citep{median1,median2}, the median voter's vote (in a binary choice) is favored by the majority.  By a careful design, our mechanism ensures that the voters who are ``below'' the median have a conflict of interest to the voters who are ``above'' the median, which makes sure less than half of the voters have an incentive to deviate and those voters can only change the outcome in the unfavorable direction by the property of median.


From a high level, our work can be understood as a revelation principal applied to plurality voting.\footnote{Loosely speaking, the revelation principal states that any outcome that can be implemented in equilibrium can also be truthfully implemented in equilibrium by having the mechanism play the equilibrium strategy on behalf of the truthful agents.}  However this view is not entirely accurate.  First, our equilibrium concept is strong Bayes Nash equilibrium while we only know that plurality voting implements the majority wish outcome in (Bayesian) Nash equilibrium~\citep{Feddersen1994}.  Second, the revelation principal requires that agents report all their knowledge.  In our case, this would include the entire prior, which is not realistic.  In contrast, our mechanisms  only require that agents report a preference and a prediction of other agents' preferences.  Such reporting requirements have previously been shown to be pragmatic~\citep{john2012measuring,prelec2017solution,hosseinisurprisingly}.  Third, our setting is different than prior work~\citep{Feddersen1994}, and this makes our results incomparable.  In particular, we deal with a discrete state space.  This difference also allows us to achieve some of our results not just in the limit, but for finite sets of agents.  

Appendix~\ref{append:comparison} contains an additional comparison of the results with those in Reference~\cite{Feddersen1994}.

\subsection{Additional Related Work} 
\label{sec:related-work}
Our work is additionally related to the recent work on incentive compatible machine learning~\citep{perote2004strategy,dekel2010incentive,chen2018strategyproof}.  In these settings, the ``social choice'' being made is a machine learning predictor where agents benefit from their point having small error with respect to the chosen predictor.  As in our setting, the information of the optimal model is distributed among the agents.  Unlike our model, the private information and the preferences of the agents essentially coincide.  

\smallskip

Information elicitation without verification, sometimes call peer prediction, is another very related line of research which shares some of the intuitions and techniques from information aggregation. The information elicitation literature has been well established in the past decades,
starting from \citet{prelec2004bayesian}'s Baysian Truth Serum and \citet{MRZ05}'s peer-prediction method.  These mechanisms cleverly design payments to the agents to guarantee that the truthful reporting of received information forms a Nash equilibrium.

A mass of recent work (see~\citet{faltings2017game} for a survey) is dedicated to designing information elicitation mechanisms that work in more general settings (such as, supporting a small number of agents~\citep{dasgupta2013crowdsourced,zhang2014elicitability,fangyi2020wine}, allowing agents having information with different levels of sophistication~\citep{2016arXiv160607042G,kong2018Eliciting}), or achieving better truthful guarantees (such as, strict Nash equilibrium~\citep{SchoenebeckY2019robust}, informed Nash equilibrium~\citep{shnayder2016informed}, or even dominant strategy equilibrium~\citep{Kong2019,kong2020}) sometimes by studying more restrictive settings (e.g. multiple similar questions being asked simultaneously~\citep{dasgupta2013crowdsourced}).
Indeed, following Bayesian Truth Serum, many of these mechanism require the agents to predict other agents' reports~\citep{zhang2014elicitability,KongS2018Monotone,kong2018Eliciting,Kong2019,kong2020aaai,fangyi2020wine}

However, all these mechanisms rely on \emph{payments} to the agents to incentivize truth-telling.  In our social choice setting, on the other hand, we need to incentivize truth-telling solely based on choosing the winning alternative.

\bt{Address the following remark for the EC reviewer:

The introduction is somewhat confusing. While in the end I understand that
you are working on a probabilistic setting, I think it might be better
position the work along epistemic social choice (and a few papers were
cited afterwards including Condorcet's work), and then introduce the
motivating example (I like the example but currently it sounds like the
authors are discovering something that is overlooked by the community---it
is never overlooked, just there are more papers on other aspects of social
choice). See the recent survey for more work in this direction (and this
paper is related to the branch on strategic voters discussed there).

Shmuel Nitzan and Jacob Paroush. Collective Decision Making and Jury
Theorems. The Oxford Handbook of Law and Economics: Volume 1: Methodology
and Concepts. 2017
}
\gs{I added the reference to the main file.  We already say that it is widely studied.}

\section{Model and Preliminaries}
\label{sect:prelim_main}
In this paper, we will define our model and present our main result with two states and two signals.  The extension to general numbers of states and signals is discussed in Appendix~\ref{append:nonbinary}.  

Suppose a department of $T$ faculty members, or agents, need to decide whether or not to hire a new faculty candidate.
In our model, those $T$ agents are voting for two \emph{alternatives}, \acc and \rej (corresponding to ``accept'' and ``reject'').
There is a set of $2$ possible \emph{worlds} (or \emph{states}) $\calW=\{L,H\}$, which describes the underlying quality of the candidate. 
Here, $L$ stands for ``low quality'' where more agents prefer \rej, and $H$ stands for ``high quality'' where more agents prefer \acc.
Agents do not know which world is the actual world that they are in.
They have a common prior belief on the likelihood of each world.
In the candidate hiring example, the CV of the candidate is given to those $T$ faculty members before any individual interviews, and a prior belief is formed.
Let $W$ be the actual world which is viewed as a random variable.
Let $(P_L,P_H)=(\Pr(W=L),\Pr(W=H))$ be the prior over worlds.
Each agent knows the values of $P_L$ and $P_H$ as prior beliefs.
We further assume $P_L,P_H>0$.

An individual interview for this candidate is held for each of the $T$ agents. 
Each agent $t$ receives a \emph{signal}, represented by the random variable  $S_t$,  from the set $\calS=\{\ell,h\}$.
Given $W=L$ or $W=H$, the signals agents receive have the same distribution and are conditionally independent.
Let $P_{\ell L}=\Pr(S_t=\ell \mid W=L)$ be the probability that signal $\ell$ will be received (by an arbitrary agent $t$) if the actual world is $L$.
Let $P_{hL},P_{\ell H}$ and $P_{hH}$ have similar meanings.
The set of values $\{P_{\ell L},P_{\ell H},P_{hL},P_{hH}\}$ is known by all the agents.
Naturally, signals are positively correlated to the worlds:
\begin{equation}\label{eqn:postive_correlation_binary}
    P_{\ell L}>P_{\ell H}\qquad\mbox{and}\qquad P_{hH}>P_{hL}.
\end{equation}
However, signals may be systematically biased. For example, it is possible that $\ell$ is more likely to be received in both worlds: $P_{\ell L}>P_{hL}$ and $P_{\ell H}>P_{hH}$.

Each agent $t$ is assigned a \emph{utility function} $v_t:\calW\times\{\acc,\rej\}\to\{0,1,\ldots,B\}$.
Naturally, voters receive higher utilities for \acc in world $H$ and for \rej in world $L$:
\begin{equation}\label{eqn:ut_natural}
    v_t(H,\acc)>v_t(L,\acc)\qquad\mbox{and}\qquad v_t(H,\rej)<v_t(L,\rej).
\end{equation}
Since we can always rescale agents' utilities, for simplicity, we assume without loss of generality that agents' utilities are integers and bounded by $B\in\mathbb{Z}^+$.
Endowed with their prior beliefs, upon receiving their signals, agents will have posterior beliefs about the distribution of $W$ and react to the mechanism in a way maximizing their expected utilities accordingly.

We assume $v_t(L,\acc)\neq v_t(L,\rej)$ and $v_t(H,\acc)\neq v_t(H,\rej)$ for each agent $t$, so that agents always strictly prefer one alternative over the other.  
Let $F$ be the set of the \emph{candidate-friendly} agents $t$ who always prefer $\acc$: $v_t(H,\acc)>v_t(L,\acc)>v_t(L,\rej)>v_t(H,\rej)$. 
Let $U$ be the set of the \emph{candidate-unfriendly} agents $t$ who always prefer $\rej$: $v_t(L,\rej)>v_t(H,\rej)>v_t(H,\acc)>v_t(L,\acc)$.
Let $C$ be the set of the \emph{contingent} agents $t$ whose preference depends on the actual world: $v_t(L,\rej)>v_t(L,\acc)$ and $v_t(H,\acc)>v_t(H,\rej)$.

Let $\alpha_F=\frac{|F|}{|T|},\alpha_U=\frac{|U|}{|T|}$ and $\alpha_C=\frac{|C|}{|T|}$ be the fractions of the three types of agents.
Since the numbers of theory, AI, software, hardware faculty members are known to everyone, we assume that the values of $\alpha_F,\alpha_U$ and $\alpha_C$ are common knowledge.
Admittedly, this assumption may not apply to some specific scenarios.
In Sect.~\ref{sect:unknowndistribution}, we discuss the model where agents have only partial information on $\alpha_F,\alpha_U$ and $\alpha_C$, and present a strong impossibility result for this model.

The goal is to output the \emph{majority wish}, the alternative that is preferred by at least half of the agents conditioned on the true state.
We assume $T$ is an odd number to avoid ties.
Clearly, \acc should be output if $\alpha_F>\frac12$, \rej should be output if $\alpha_U>\frac12$.
In the case $\alpha_F,\alpha_U<\frac12$, $\acc$ should be output if the actual world is $H$ and $\rej$ should be output if the actual world is $L$.

Our results will sometimes require $T$, the number of agents, to be sufficiently large, and it may be helpful to think of $T\rightarrow\infty$.  However, we will always assume that the parameters of the model: $B$, $\{P_L,P_H\}$,  $\{P_{\ell L},P_{\ell H},P_{hL},P_{hH}\}$, and $\{\alpha_F,\alpha_U,\alpha_C\}$, do not depend on $T$ in any way.

The traditional social choice setting with agents having predetermined preferences can be viewed as a special case of our model, by setting $|C|=0$ (i.e., there is no contingent agent).

For the ease of comprehension, we have used the faculty candidate hiring as a running example for this paper.  This can be replaced by any example from most practical scenarios where different types of imperfectly informed voters are voting between two alternatives, including all the examples we mentioned in Section~\ref{sect:intro}.

In our election example in the second paragraph of Sect.~\ref{sect:intro}, \acc and \rej can represent candidates $a$ and $b$ respectively.
Correspondingly, for this example, $L$ and $H$ can represent ``$a$ is better suited'' and ``$b$ is better suited'' respectively.  $F$ and $U$ represent voters aligned to party $A$ and $B$ respectively, while $C$ represents those swing voters whose preferences depend on the signals (in this case, the signals correspond to their private experiences and beliefs, which may be based on information they obtained from their favorite TV programs, newspapers, etc).

As a remark, in the main body of this paper, we discuss the scenario with two worlds and two signals.
This captures many essential ideas behind our mechanism for general cases, and we view the result in this section as the main contribution of this paper.
In general, we can have the set of worlds be $\calW=\{1,\ldots,N\}$ which describes the quantitative quality, and we can have the set of signals be $\calS=\{1,\ldots,M\}$ which are positively correlated to the worlds (see Appendix~\ref{append:nonbinary}).
However, in many scenarios, it is reasonable and also more practical to assume that the worlds and the signals can only be either ``good'' or ``bad''. 
This setting is the simplest, cleanest, and, in many ways, the most intuitive.  Additionally, it is much easier for different agents to be consistent in distinguishing good and bad than to similarly evaluate the quality using numerical scale from $1$ to $M$.
In addition, numerical ranking causes more subjective systematic bias, and the heterogeneity of the bias among the agents makes agents' reports more noisy.
Finally, our mechanism can be much easier to implement under the setting with $N=2$ and $M=2$, which makes our mechanism more appealing in practice.

\subsection{Strategy and \texorpdfstring{$\varepsilon$}{epsilon}-strong Bayes Nash equilibrium}
A \emph{mechanism} collects a \emph{report} from each agent, and then outputs an alternative which is either \acc or \rej.
The mechanism specifies the content of the report by specifying questions for the agents.
Examples of those questions include asking each agent for the signal he/she receives, asking each agent to predict the other agents' reports, etc.

Let $\mathcal{R}$ be the space of all possible reports, which depends on the design of the mechanism.
A \emph{pure strategy} of an agent is given by a function $\sigma:\calS\to\mathcal{R}$ that maps a signal received by this agent to a report.
In a \emph{mixed strategy}, $\sigma$ can be a random function.

An agent's strategy is \emph{truthful} if it always specifies the correct answer to each question in the report, to the best of the agent's knowledge after receiving the signal.
For example, if the mechanism asks for the agent's signal, an agent playing the truthful strategy should report the signal he/she receives; if the mechanism asks the agents to predict the fraction of agents who will receive signal $m$, an agent playing the truthful strategy should report his/her posterior belief on this computed by the Bayes rule (Section~\ref{sect:posterior_update_main} discusses the computation of posterior beliefs). 

Given a strategy profile $\Sigma=(\sigma_1,\ldots,\sigma_T)$, let $u_t(\Sigma)$ be the expected utility of agent $t$, where the expectation is taken over the sampling of agents' signals.
Notice that we use $u_t$ to denote the \emph{ex-ante} utility (as defined just now) and we have used $v_t$ to denote the \emph{ex-post} utility (see the fourth paragraph in Section~\ref{sect:prelim_main}). 
Most parts of this paper will focus on the \emph{ex-ante} utility, especially when we are talking about any equilibrium solution concept.

Since we are in a social choice setting with a potentially large number of agents, a single agent's behavior may not have much effect.  Thus, instead of the typical Bayes Nash Equilibrium, we consider a much stronger goal---the strong Bayes Nash equilibrium.

\begin{definition}\label{def:SBNE}
A strategy profile $(\sigma_1,\ldots,\sigma_T)$ is an \emph{$\varepsilon$-strong Bayes Nash equilibrium} if there does not exist a subset of agents $D$ and a strategy profile $(\sigma_1',\ldots,\sigma_T')$ such that
\begin{enumerate}
    \item $\sigma_t=\sigma_t'$ for each $t\notin D$,
    \item $u_t(\sigma_1',\ldots,\sigma_T')\geq u_t(\sigma_1,\ldots,\sigma_T)$ for each $t\in D$, and
    \item there exist $t\in D$ such that $u_t(\sigma_1',\ldots,\sigma_T')> u_t(\sigma_1,\ldots,\sigma_T)+\varepsilon$.
\end{enumerate}
\end{definition}

\gs{I put in more things discribing the equilibrium}
When a strategy profile $(\sigma_1,\ldots,\sigma_n)$ is not an $\varepsilon$-strong Bayes Nash equilibrium, we will call the subset of the agents $D$ in Definition~\ref{def:SBNE} the \emph{deviating agents} or the \emph{deviating coalition}. By Definition~\ref{def:SBNE}, every agent of the deviating coalition must be at least as well off, and some must be strictly better off by at least $\epsilon$.
Notice that a $0$-strong Bayes Nash equilibrium is the conventional strong Bayes Nash equilibrium.  The larger $\varepsilon$ is, the harder it is to find a deviating coalition, and so the larger the set of $\varepsilon$-strong Bayes Nash equilibria.

\subsection{Posterior Update by Bayes Rule}
\label{sect:posterior_update_main}

\gs{Why do we ever have $\mathcal{M}$ in the $\lambda$ notation?  Isn't the mechanism always clear?}

Upon receiving a signal $S_t\in\{\ell,h\}$, agent $t$ updates his/her posterior beliefs (about the probability that (s)he is in world $L$ or $H$, the fraction of agents that will receive signal $\ell$ or $h$, etc.) based on Bayes rule.
Let $T_{m'm}$ be the probability that an agent who receives signal $m\in\{\ell,h\}$ believes that another agent will receive signal $m'\in\{\ell,h\}$.

Suppose agent $t$ receives signal $S_t=m\in\{\ell,h\}$.
(S)he believes that the actual world is $n\in\{L,H\}$ with probability
$$\Pr\left(W=n\mid S_t=m\right)=\frac{\Pr(W=n)\Pr(S_t=m\mid W=n)}{\Pr(S_t=m)}=\frac{P_nP_{mn}}{P_LP_{mL}+P_HP_{mH}}.$$
Then, $T_{m'm}$ can be computed as follows:
\begin{align}\label{eqn:Tmm'_main}
    T_{m'm}&=\Pr\left(W=L\mid S_t=m\right)\cdot P_{m'L}+\Pr\left(W=H\mid S_t=m\right)\cdot P_{m'H}\nonumber\\
    &=\frac{P_LP_{mL}}{P_LP_{mL}+P_HP_{mH}}\cdot P_{m'L}+\frac{P_HP_{mH}}{P_LP_{mL}+P_HP_{mH}}\cdot P_{m'H}.
\end{align}


%




Given a strategy profile $\Sigma=\{\sigma_1,\ldots,\sigma_T\}$ and a mechanism $\mathcal{M}$, let $\lambda_n^\mathcal{\acc, M}(\Sigma)$ be the probability that alternative \acc is announced as the winner given the actual world is $n$, then $\lambda_n^\mathcal{\rej, M}(\Sigma)=1-\lambda_n^\mathcal{\acc, M}(\Sigma)$ is the probability that alternative \rej wins given the actual world is $n$.
We will omit the superscript $\mathcal{M}$ when it is clear what mechanism we are discussing.

All the agents' \emph{ex-ante} utilities depend exclusively on $\lambda_L^\acc(\Sigma),\lambda_H^\acc(\Sigma)$ (or $\lambda_L^\rej(\Sigma),\lambda_H^\rej(\Sigma)$), and each agent $t$'s utility is given by
\begin{equation}\label{eqn:utsigma1_main}
    u_t(\Sigma)=P_L\left(\lambda_L^\acc(\Sigma)v_t(L,\acc)+\lambda_L^\rej(\Sigma)v_t(L,\rej)\right)+P_H\left(\lambda_H^\acc(\Sigma)v_t(H,\acc)+\lambda_H^\rej(\Sigma)v_t(H,\rej)\right),
\end{equation}
which can also be rewritten as (by noticing $\lambda_L^\acc(\Sigma)+\lambda_L^\rej(\Sigma)=1$ and $\lambda_H^\acc(\Sigma)+\lambda_H^\rej(\Sigma)=1$)
\begin{equation}\label{eqn:utsigma2_main}
    u_t(\Sigma)=P_Lv_t(L,\acc)+P_Hv_t(H,\rej)+P_L\lambda_L^\rej(\Sigma)(v_t(L,\rej)-v_t(L,\acc))+P_H\lambda_H^\acc(\Sigma)(v_t(H,\acc)-v_t(H,\rej)),
\end{equation}
and
\begin{equation}\label{eqn:utsigma3_main}
     u_t(\Sigma)=P_Lv_t(L,\rej)+P_Hv_t(H,\rej)+P_L\lambda_L^\acc(\Sigma)(v_t(L,\acc)-v_t(L,\rej)) 
     +P_H\lambda_H^\acc(\Sigma)(v_t(H,\acc)-v_t(H,\rej)).
\end{equation}
We will always use $\Sigma^\ast=\{\sigma_1^\ast,\ldots,\sigma_T^\ast\}$ to denote the truthful strategy profile.

Table~\ref{tab:notation_main} lists all the frequently used notations.
\begin{table}[ht]
    \centering
    \begin{tabular}{ll}
    \hline
    notation & meaning\\
    \hline
        $T$ & the total number of agents\\
        $\calW=\{L,H\}$ & the set of all worlds \\
        $\calS=\{\ell,h\}$ & the set of all signals\\
        $P_n$ & the prior belief for the probability the actual world is $n$\\
        $P_{mn}$ & the probability of receiving signal $m$ under world $n$\\
        $T_{m'm}$ & the posterior belief for another agent to receive $m'$ given that signal $m$ is received\\
        $v_t(n,\acc),v_t(n,\rej)$ & the (\emph{ex-post}) utility for agent $t$ for alternative \acc, \rej if the actual world is $n$\\
        $B$ & the upper bound for all agents' (\emph{ex-post}) utilities\\
        $u_t(\Sigma)$ & the (\emph{ex-ante}) expected utility for agent $t$ given strategy profile $\Sigma$\\
        $F,C,U$ & candidate-friendly agents, contingent agents, candidate-unfriendly agents\\
        $\alpha_F,\alpha_C,\alpha_U$ & fractions of the three types of agents\\
        $\lambda_n^\acc(\Sigma),\lambda_n^\rej(\Sigma)$ & the probability a given mechanism outputs \acc, \rej for strategy profile $\Sigma$\\
        $\Sigma^\ast$ & the truthful strategy profile\\
        $c$ & will be defined in (\ref{eqn:C})\\
        $I(\Sigma)$ & under the assumptions $\alpha_F<0.5$ and $\alpha_U<0.5$, the probability a given mechanism\\
        & outputs the alternative \emph{not} favored by the majority given the strategy profile $\Sigma$;\\
        & will be defined in (\ref{eqn:errorrate})\\
        \hline
    \end{tabular}
    \caption{Table of notations.}
    \label{tab:notation_main}
\end{table}

\section{The Wisdom-of-the-Crowd-Voting Mechanism}
\label{sect:twoworldstwosignals}
We will first review Prelec et al.'s \emph{Surprisingly Popular} algorithm \citep{prelec2017solution}, which works under a setting similar to ours but with non-strategic agents.
Some part of the intuition behind our mechanism is based on Prelec et al.'s work.

\subsection{Prelec et al.'s Surprisingly Popular Algorithm}
\label{sect:prelec_et_al}
For the purpose of this paper, we will describe the algorithm with two worlds and two signals.
The algorithm asks each agent $t$ the signal (s)he receives, and his/her belief on the fraction of agents who have received signal $\ell$ (or signal $h$).
In our notation, each agent reports the realization of $S_t$ and, assuming $S_t=m\in\{\ell,h\}$, the value of $T_{\ell m}$ (or $T_{hm}$, which equals to $1-T_{\ell m}$).
Since agents are assumed to be non-strategic, those who receive signal $\ell$ will report $(\ell,T_{\ell\ell})$ and those who receive signal $h$ will report $(h,T_{\ell h})$.
The algorithm then computes the fraction of agents who report signal $\ell$, and the average value of all the reported $T_{\ell m}$'s.
If the former is greater than the latter, $\ell$ is considered as being ``surprisingly popular'' and the algorithm will conclude that $L$ is the actual world.
Otherwise, $h$ will be considered as being ``surprisingly popular'' and $H$ will be concluded as being the actual world.

The correctness of this algorithm is based on the following simple yet important observation in Theorem~\ref{thm:keyInequality}.
In particular, the average of agents' reported predictions (those $T_{\ell m}$'s) will be between $T_{\ell h}$ and $T_{\ell\ell}$.
When the number of agents $T$ is sufficiently large, the actual fraction of agents who receive signal $\ell$ will be either approximately $P_{\ell L}$ (if $L$ is the actual world) or approximately $P_{\ell H}$ (if $H$ is the actual world).
Theorem~\ref{thm:keyInequality} then implies the correctness of the Surprisingly Popular algorithm.

\begin{theorem}\label{thm:keyInequality}
    $P_{\ell H}<T_{\ell h}<T_{\ell\ell}<P_{\ell L}$ and $P_{hH}>T_{hh}>T_{h \ell}>P_{h L}$
\end{theorem}   
The intuition behind the theorem is straightforward. 
The inequality $P_{\ell H}<P_{\ell L}$ is by the positive correlation~(\ref{eqn:postive_correlation_binary}).
The inequality $T_{\ell h}<T_{\ell\ell}$ is also intuitive: the positive correlation between the signals and worlds implies the positive correlation between two agents' received signals.
Finally, (\ref{eqn:Tmm'_main}) implies each of $T_{\ell h}$ and $T_{\ell\ell}$ is a weighted average of $P_{\ell H}$ and $P_{\ell L}$, so the value is between $P_{\ell H}$ and $P_{\ell L}$.
This concludes the first inequality chain, and the second can be shown similarly.
\begin{proof}
We will only show the first inequality chain.  The second chain follows directly from the first by noticing each term in the second chain is 1 minus a term in the first.  

By (\ref{eqn:Tmm'_main}) and $P_{\ell H}<P_{\ell L}$ in (\ref{eqn:postive_correlation_binary}), we have $$T_{\ell\ell}=\frac{P_LP_{\ell L}^2+P_HP_{\ell H}^2}{P_LP_{\ell L}+P_HP_{\ell H}}<\frac{P_LP_{\ell L}^2+P_HP_{\ell H}P_{\ell L}}{P_LP_{\ell L}+P_HP_{\ell H}}=P_{\ell L}$$
and 
$$T_{\ell h}=\frac{P_LP_{\ell L}P_{hL}+P_HP_{\ell H}P_{hH}}{P_LP_{hL}+P_HP_{hH}}>\frac{P_LP_{\ell H}P_{hL}+P_HP_{\ell H}P_{hH}}{P_LP_{hL}+P_HP_{hH}}=P_{\ell H}.$$
Finally, to show $T_{\ell\ell}>T_{\ell h}$, it suffices to show that 
$$\pi_1:=\frac{P_LP_{\ell L}}{P_LP_{\ell L}+P_HP_{\ell H}}>\pi_2:=\frac{P_LP_{hL}}{P_LP_{hL}+P_HP_{hH}},$$ 
since $T_{\ell\ell}=\pi_1P_{\ell L}+(1-\pi_1)P_{\ell H}$, $T_{\ell h}=\pi_2P_{\ell L}+(1-\pi_2)P_{\ell H}$ and $P_{\ell L}>P_{\ell H}$. 
Simple calculations show this:
$$\pi_1>\frac{P_L}{P_L+P_H}>\pi_2,$$
where the first inequality is due to $P_{\ell L}>P_{\ell H}$ and the second inequality is due to $P_{hH}>P_{hL}$.
\end{proof}

Throughout this section, we use $c$ to denote the following constant.
\begin{equation}\label{eqn:C}
    c=\frac13\min\left\{T_{\ell h}-P_{\ell H},T_{\ell\ell}-T_{\ell h},P_{\ell L}-T_{\ell\ell}, P_{hH}-T_{hh},T_{hh}-T_{\ell h}, T_{\ell h}-P_{\ell H}\right\}
\end{equation}

\subsection{The Wisdom-of-the-Crowd-Voting Mechanism}
\label{subsec:mech}

At Step~3 of the mechanism, we only elicit predictions from those contingent agents, and those candidate-friendly (candidate-unfriendly resp.) agents' predictions are treated as $0$ ($1$ resp.).
In Sect.~\ref{sect:alternative_mechanism}, we discuss an alternative mechanism where we elicit all the predictions and then take the median.
The alternative mechanism shares the same theoretical properties, and we discuss the advantages and disadvantages of this alternative.

\gs{I change the mechanism so that the predictions are changed as well.  This make it more "strategy-proof"}
\bt{Sure. Do you think it is a good idea that we also remark that we do not need to change the predictions? In my opinion, changing the predictions make the proof goes easier. However, it will make our mechanism appears stranger to agents without background. An agent can understand why his signal is changed to his favored candidate, but he may not see why having an extreme prediction helps.}
\bt{After some tries, I decided not to convert the predictions. My first reason is that, if we convert the predictions, we need may be two more sentences for the proof of Theorem~\ref{thm:msw} (it is slightly harder to argue that the median still falls into the middle interval), and that will make our paper exceeds the page limit. For the second reason, I still do not like the idea of this conversion, for the reasons mentioned in my previous comment.}

Mechanism~\ref{mechanism:M=N=2} may look too obscure to be implemented in practice.
However, very simple and understandable questionnaires implementing the mechanism can be designed.
In our running example of faculty candidate hiring, the questionnaire corresponding to our mechanism could look like:

\begin{enumerate}
    \item Choose one of the followings:
    \begin{enumerate}
        \item I definitely want to accept this candidate, regardless of my colleagues' inputs.
        \item I definitely want to reject this candidate, regardless of my colleagues' inputs.
        \item After talking to the candidate, I am more inclined to accept him/her than before.
        \item After talking to the candidate, I am more inclined to reject him/her than before.
    \end{enumerate}
    \item If your answer is (c) or (d) in the first question, what percentage of the faculty members do you believe will choose (a) or (c) in the first question?
\end{enumerate}

\begin{algorithm}[H]
\caption{The Wisdom-of-the-Crowd-Voting Mechanism}
\label{mechanism:M=N=2}
\begin{algorithmic}[1]
    \STATE Each agent $t$ reports his/her type ($F$, $U$ or $C$) to the mechanism, and if he/she is of type $C$, the signal (s)he receives (either $\ell$ or $h$), denoted by $\bar{s}_i\in\{\ell,h\}$. 
    \STATE If agent $t$ reports type $F$, his reported signal will be automatically treated as $\bar{s}_t=h$; 
    if agent $t$ reports type $U$, his reported signal will be automatically treated as $\bar{s}_t=\ell$. 
    \STATE For each agent $t$ of type $C$, ask him/her to predict the fraction of agents who will report signal $h$. Let $\bar{\delta}_t$ be $t$'s prediction. \emph{The prediction $\bar{\delta}_t$ should be made with the type $F$ and type $C$ agents' predictions defined in the previous step being considered, and the mechanism makes this clear to the agents.} For each agent $t$ of type $F$, $\bar{\delta}_t$ is set to $0$, and for each agent $t$ of type $U$, $\bar{\delta}_t$ is set to $1$.
    \STATE Compute the \emph{median} of those $\bar{\delta}_t$, denoted by $\bar\delta$.
    \STATE If more than half of the agents report type $F$, announce \acc being the winning alternative; if more than half of the agents reports type $U$, announce \rej being the winning alternative.
    \STATE If the fraction of the agents reporting $\bar{s}_t=h$ is more than the median $\bar\delta$, announce \acc being the winning alternative; otherwise, announce \rej being the winning alternative.
\end{algorithmic}
\end{algorithm}

\bt{the figure doesn't look good to me...} \gs{you need to store it in a format that does not lose information, such as eps or pdf}

In the first question above, faculty members of type $F$ (type $U$ resp.) will choose (a) ((b) resp.).
Faculty members of type $C$ will choose either (c) or (d) depending on the signals they have received.
If a faculty member receives signal $h$, (s)he believes world $H$ is more likely than before.
Notice that, it is still possible that (s)he believes the probability of world $H$ being the actual world is less than $50\%$ and (s)he still prefers rejecting the candidate based on his/her private information (for example, his/her prior belief may be only $10\%$ for world $H$, and his/her posterior belief for this increases to $30\%$ upon receiving a signal $h$), so the description that (s)he is ``more inclined to accept the candidate than before'' accurately implements Mechanism~\ref{mechanism:M=N=2}.
The same holds for those receiving signal $\ell$.

\subsection{Main Theoretical Results for Our Mechanism}
We first show that our mechanism indeed achieves (with an exponentially small failure probability) the goal of outputting the alternative favored by the majority, assuming agents are truth-telling.  


\begin{theorem}\label{thm:msw}
If all the agents play the truthful strategy $\Sigma^*$, then, with probability at least $1-2\exp(-2c^2\alpha_{C}T)$ (refer to Table~\ref{tab:notation_main} for notations), our mechanism outputs an alternative that is favored by more than half of the agents.
\end{theorem}
\begin{proof}
Suppose all the agents report truthfully.
Step~5 of the mechanism guarantees that the majority wish will be announced with probability $1$ if either $\alpha_F>0.5$ or $\alpha_U>0.5$.
It remains to consider the case where we have both $\alpha_F<0.5$ and $\alpha_U<0.5$.\footnote{Recall that we have assume $T$ is an odd number, so we cannot have $\alpha_F=0.5$ or $\alpha_U=0.5$.}
In this case, \acc is favored by the majority if the actual world is $H$, and \rej is favored by the majority if the actual world is $L$.

If a contingent agent $t$ receives signal $h$, (s)he will believe that a $T_{hh}$ fraction of agents receive $h$ (before the treatment at Step~2), and (s)he will report $\bar{\delta}_t=\alpha_{C}T_{hh}+\alpha_{F}$ (after considering the treatment at Step~2).
Similarly, a contingent agent receiving signal $\ell$ will report $\bar{\delta}_t=\alpha_{C}T_{h\ell}+\alpha_{F}$.
We have $\bar{\delta}_t=0$ for each candidate-friendly agent and $\bar{\delta}_t=1$ for each candidate-unfriendly agent.
Since we are considering the case $\alpha_F,\alpha_U<0.5$, the median $\bar\delta$ is in the interval $[\alpha_{C}T_{h\ell}+\alpha_{F},\alpha_{C}T_{hh}+\alpha_{F}]$ (note that $T_{h\ell}<T_{hh}$ by Theorem~\ref{thm:keyInequality}).

Suppose the actual world is $L$.
The expected fraction of the agents receiving signal $h$ would be $P_{hL}$, and the expected fraction of the agents reporting signal $h$ (after the treatment at Step~2) would be $\alpha_{C}\cdot P_{hL}+\alpha_{F}$. 
By a Chernoff bound, with probability at least $1-2\exp(-2c^2\alpha_{C}T)$, the fraction of agents reporting signal $h$ is in the interval $[\alpha_{C}\cdot (P_{hL}-c)+\alpha_{F},\alpha_{C}\cdot
(P_{hL}+c)+\alpha_{F}]$, which is less than $\alpha_{C}\cdot T_{h\ell}+\alpha_{F}\leq\bar{\delta}$ by Theorem~\ref{thm:keyInequality} and (\ref{eqn:C}).
Step~6 of the mechanism indicates that \rej will be announced.
The analysis for the case where $H$ is the actual world is similar.
\end{proof}

Next, we show that the truthful strategy profile is an $\varepsilon$-strong Bayes Nash equilibrium of our mechanism for some exponentially small $\varepsilon$.

\begin{theorem}\label{thm:SBNE2}
The truthful strategy profile is an $\varepsilon$-strong Bayes Nash equilibrium, where $\varepsilon=(2B^2+4B)\exp(-2c^2\alpha_{C}T)$.
\end{theorem}

\subsection{Proof Theorem~\ref{thm:SBNE2} with \texorpdfstring{$T\rightarrow\infty$}{T to infinity}}  
We defer the full proof of Theorem~\ref{thm:SBNE2} to Sect.~\ref{sect:omitted_proofs_mainTheorem}.  Here we prove the following limit version, where $T\rightarrow\infty$, which illustrates the key features while eliminating the need for both Chernoff bound analyses and some additional subtle corner cases. 

\begin{theorem}\label{thm:SBNE_limit}
For $T\rightarrow\infty$, the truthful strategy profile is a strong Bayes Nash equilibrium.
\end{theorem}

We consider three cases: 1) $\alpha_{F}>0.5$, 2) $\alpha_{U}>0.5$ and 3) $\alpha_{F}<0.5$ and $\alpha_{U}<0.5$.

In the first case, more than half of the agents are candidate-friendly, and \acc will be announced according to Step~5 of the mechanism if these agents report truthfully.
The truthful strategy profile forms a strong Bayes Nash equilibrium, as those candidate-friendly agents receive their maximum utilities by truth-telling and the remaining agents are not able to stop the mechanism from outputting \acc regardless of what they report.
The analysis for the second case is the analogous to the first case.
It remains to consider the third case.

Under the third case, \acc is favored by the majority if the actual world is $H$, and \rej is favored by the majority if the actual world is $L$.
By the same analysis as in the proof of Theorem~\ref{thm:msw}, supposing agents report truthfully, we know that \acc will be output with probability $1$ (by taking the limit $T\rightarrow\infty$) if the actual world is $H$, and \rej will be output if the actual world is $L$.
Therefore, the contingent agents (type $C$) receive their maximum utilities, and thus have no incentive to deviate from the truthful strategy.

To conclude that truth-telling is a strong Bayes Nash equilibrium, we will show that there is no coalition of deviating agents $D$ (see the paragraph following Definition~\ref{def:SBNE}).
Let $\Sigma'$ be the strategy profile after $D$'s deviation.

Next, we show that $D$ cannot contain both a type $F$ agent and a type $U$ agent.
The following lemma shows that an increase in a type $F$ agent's (\emph{ex-ante}) utility always results a decrease in a type $U$ agent's (\emph{ex-ante}) utility, and vice versa.
This is obvious if we are dealing with \emph{ex-post} utilities, as a candidate-friendly agent and a candidate-unfriendly agent always want the opposite alternatives.
However, this becomes less obvious for \emph{ex-ante} utilities.

\begin{lemma}\label{lem:conflict_limit}
Suppose $\alpha_{F}<0.5$ and $\alpha_{U}<0.5$.
Let $\Sigma^\ast$ be the truthful strategy profile and $\Sigma'$ be an arbitrary strategy profile.
Let $t_1$ be an arbitrary candidate-friendly agent and $t_2$ be an arbitrary candidate-unfriendly agent.
We have
\begin{description}
    \item[(i)] If $u_{t_1}(\Sigma')-u_{t_1}(\Sigma^\ast)>0$, then $u_{t_2}(\Sigma')-u_{t_2}(\Sigma^\ast)<0$.
    \item[(ii)] If $u_{t_2}(\Sigma')-u_{t_2}(\Sigma^\ast)>0$, then $u_{t_1}(\Sigma')-u_{t_1}(\Sigma^\ast)<0$.
\end{description}
\end{lemma}
\begin{proof}
By (\ref{eqn:utsigma3_main}), we have 
$$u_t(\Sigma')-u_t(\Sigma^\ast)=\Gamma_L(v_t(L,\acc)-v_t(L,\rej))-\Gamma_H(v_t(H,\acc)-v_t(H,\rej)),$$ 
where $\Gamma_L=P_L(\lambda_L^\acc(\Sigma')-\lambda_L^\acc(\Sigma^\ast))$ and $\Gamma_H=P_H(\lambda_H^\acc(\Sigma^\ast)-\lambda_H^\acc(\Sigma'))$.  
Theorem~\ref{thm:msw} implies $\lambda_L^\acc(\Sigma^\ast)=0$ and $\lambda_H^\acc(\Sigma^\ast)=1$ (with $T\rightarrow\infty$), which implies $\Gamma_L,\Gamma_H\geq0$.
By (\ref{eqn:ut_natural}), we have $v_t(L,\acc)-v_t(H,\acc)<0<v_t(L,\rej)-v_t(H,\rej)$, which further implies $v_{t}(L,\acc)-v_{t}(L,\rej)<v_{t}(H,\acc)-v_{t}(H,\rej)$.
Since a type $F$ agent always prefers $\acc$ and a type $U$ agent always prefers $\rej$,
\begin{eqnarray*}
0<v_{t_1}(L,\acc)-v_{t_1}(L,\rej)&<&v_{t_1}(H,\acc)-v_{t_1}(H,\rej),\qquad\qquad\mbox{ and }\\
v_{t_2}(L,\acc)-v_{t_2}(L,\rej)&<&v_{t_2}(H,\acc)-v_{t_2}(H,\rej)<0.
\end{eqnarray*}
Intuitively, this means $u_{t_1}(\Sigma')-u_{t_1}(\Sigma^\ast)$ is more sensitive to $\Gamma_H$ and $u_{t_2}(\Sigma')-u_{t_2}(\Sigma^\ast)$ is more sensitively to $\Gamma_L$.
Formally, $u_{t_1}(\Sigma')-u_{t_1}(\Sigma^\ast)>0$ implies $\Gamma_L>\Gamma_H\frac{v_{t_1}(H,\acc)-v_{t_1}(H,\rej)}{v_{t_1}(L,\acc)-v_{t_1}(L,\rej)}\geq\Gamma_H$, and  $u_{t_2}(\Sigma')-u_{t_2}(\Sigma^\ast)\geq0$ implies $\Gamma_L\leq\Gamma_H\frac{v_{t_2}(H,\acc)-v_{t_2}(H,\rej)}{v_{t_2}(L,\acc)-v_{t_2}(L,\rej)}\leq\Gamma_H$.
Thus, $u_{t_1}(\Sigma')-u_{t_1}(\Sigma^\ast)>0$ and $u_{t_2}(\Sigma')-u_{t_2}(\Sigma^\ast)\geq0$ cannot be both true.
The proof for (ii) is similar.
\end{proof}

We have seen that the deviating coalition $D$ cannot contain any contingent agents (since their utilities have already been maximized).
Thus, Lemma~\ref{lem:conflict_limit} implies $D$ can only be comprised of either candidate-friendly agents or candidate-unfriendly agents.
Finally, we show that a minority coalition comprised of only candidate-friendly agents or only candidate-unfriendly agents cannot change the outcome by misreporting.
We consider candidate-friendly agents without loss of generality.

Suppose $D$ contains only candidate-friendly agents.
Those agents cannot make the mechanism output \acc at Step~5, because fewer than $1/2$ of the agents report type $F$ no matter how agents in $D$ deviate.
To maximize the chance for the mechanism to output \acc at Step~6, those candidate-friendly agents would like to maximize the fraction of agents reporting $h$ and minimize the median of the prediction $\bar\delta$.
However, the mechanism already does this when those candidate-friendly agents play the truthful strategy: their signals are treated as $h$ at Step~2, and their predictions are treated as $0$ at Step~3.
The same analysis works when $D$ contains only candidate-unfriendly agents.
This concludes the proof for Theorem~\ref{thm:SBNE_limit}.

In fact, the arguments in the previous paragraph show that: \emph{truth-telling is a dominant strategy for each candidate-friendly agent and each candidate-unfriendly agent}.

\subsection{Proof of Theorem~\ref{thm:SBNE2}}
In this section, we formally prove Theorem~\ref{thm:SBNE2}.

\label{sect:omitted_proofs_mainTheorem}
\subsubsection{Proof Sketch for Theorem~\ref{thm:SBNE2}}

We first describe a sketch of the proof.
We consider three cases: 1) $\alpha_{F}>0.5$, 2) $\alpha_{U}>0.5$ and 3) $\alpha_{F}<0.5$ and $\alpha_{U}<0.5$.

\smallskip
\paragraph{Case 1) $\alpha_{F}>0.5$.} 
For the first case, more than half of the agents are candidate-friendly, and \acc will be announced according to Step~5 of the mechanism if these agents report truthfully.
The truthful strategy profile forms a ($0$-)strong Bayes Nash equilibrium, as those candidate-friendly agents receive their maximum utilities by truth-telling and the remaining agents are not able to stop the mechanism from outputting \acc regardless of what they report.

\paragraph{Case 2) $\alpha_{U}>0.5$.} 
The analysis for the second case is the analogous to the first.

\paragraph{Case 3) $\alpha_{F}<0.5$ and $\alpha_{U}<0.5$.} 
Under the third case, \acc is favored by the majority if the actual world is $H$, and \rej is favored by the majority if the actual world is $L$.
Given a strategy profile $\Sigma$, we define 
\begin{equation}\label{eqn:errorrate}
I(\Sigma) \equiv P_L\lambda_L^\acc(\Sigma)+P_H\lambda_H^\rej(\Sigma)
\end{equation} 
to be the error rate of a strategy: the probability the mechanism selects the alternative that is not the majority wish.

By Theorem~\ref{thm:msw}, we know that $I(\Sigma^\ast)\leq 2\exp(-2c^2\alpha_{C}T)$.
Specifically, supposing agents report truthfully, we know that \acc will be output with probability at least $1-2\exp(-2c^2\alpha_{C}T)$ if the actual world is $H$ (i.e., $\lambda_H^\acc(\Sigma^\ast)\geq1-2\exp(-2c^2\alpha_{C}T)$), and \rej will be output with probability at least $1-2\exp(-2c^2\alpha_{C}T)$ if the actual world is $L$ (i.e., $\lambda_L^\rej(\Sigma^\ast)\geq1-2\exp(-2c^2\alpha_{C}T)$).

To conclude that truth-telling is an $\varepsilon$-strong Bayes Nash equilibrium, we will study two cases and show that  no coalition of deviating agents $D$ exists in either case. Recall that in a deviating coalition all agents must benefit and some agent must benefit by at least  $\varepsilon=(2B^2+4B)\exp(-2c^2\alpha_CT)$.
Let $\Sigma'$ be the strategy profile after $D$'s deviation.
We consider two sub-cases:
\begin{enumerate}
    \item $I(\Sigma') < (2B+2)\exp(-2C^2\alpha_{C}T)$, and
    \item $I(\Sigma') \geq (2B+2)\exp(-2C^2\alpha_{C}T)$.
\end{enumerate}

In the first case, $I(\Sigma')$ is small, so the mechanism nearly always chooses the majority wish under $\Sigma'$. 
Therefore, the output of the mechanism does not change with high probability from profile $\Sigma^\ast$ to profile $\Sigma'$.
In particular, we have $\lambda_H^\acc(\Sigma')\approx\lambda_H^\acc(\Sigma^\ast)\approx1$ and $\lambda_L^\rej(\Sigma')\approx\lambda_L^\rej(\Sigma^\ast)\approx1$.
Thus, no agent can be much better off because all the agents have nearly the same utilities as before.
This is formally proved in Claim~\ref{claim:sameasbefore}.

In the second case, $I(\Sigma')$ is not small, so the mechanism sometimes fails to choose the majority wish under $\Sigma'$.  
Here, the contingent agents receive strictly lower utilities (Claim~\ref{claim:contingentunhappy}) and thus no contingent agent can be in the deviating coalition.  
The technical key is then to show that $D$ cannot contain both a candidate-friendly and a candidate-unfriendly agent.
Lemma~\ref{lem:conflict} shows that a significant increase in a candidate-friendly agent's (\emph{ex-ante}) utility always results a decrease in a candidate-unfriendly agent's (\emph{ex-ante}) utility, and vice versa.
This is obvious if we are dealing with \emph{ex-post} utilities, as a candidate-friendly agent and a candidate-unfriendly agent always want the opposite alternatives.
However, this becomes much less obvious for \emph{ex-ante} utilities.

Thus, any deviating coalition can only be comprised of either candidate-friendly agents or candidate-unfriendly agents.
Finally, in Claim~\ref{claim:nochange}, we show that a minority coalition comprised of only candidate-friendly agents (or only candidate-unfriendly agents) cannot change the outcome by misreporting.
This concludes the theorem.

\subsubsection{No Win-win Lemma}

A key part of the proof is the following lemma which states that it is impossible that predetermined agents of different alternatives both gain by deviating from the truthful strategy profile.  As a corollary, any deviating coalition can only contain  predetermined agents of one type.  

The proof of Lemma~\ref{lem:conflict} depends on i) truth-telling nearly always selecting the majority wish; and ii) the monotonicity of the \emph{ex-post} utilities $v_t(\cdot,\acc)$ and $v_t(\cdot,\rej)$ in the first argument (here, we mean $v_t(L,\acc)<v_t(H,\acc)$ and $v_t(L,\rej)>v_t(H,\rej)$).  In particular, Lemma~\ref{lem:conflict} does not hold if truth-telling is replaced by an arbitrary strategy profile.

\begin{lemma}\label{lem:conflict}
Suppose $\alpha_{F}<0.5$ and $\alpha_{U}<0.5$.
Let $\Sigma^\ast$ be the truthful strategy profile and $\Sigma'$ be an arbitrary strategy profile.
Let $t_1$ be an arbitrary candidate-friendly agent and $t_2$ be an arbitrary candidate-unfriendly agent.
For any $\Delta\geq2B\exp(-2c^2\alpha_{C}T)$, we have
\begin{description}
    \item[(i)] If $u_{t_1}(\Sigma')-u_{t_1}(\Sigma^\ast)>\Delta$, then $u_{t_2}(\Sigma')-u_{t_2}(\Sigma^\ast)<0$.
    \item[(ii)] If $u_{t_2}(\Sigma')-u_{t_2}(\Sigma^\ast)>\Delta$, then $u_{t_1}(\Sigma')-u_{t_1}(\Sigma^\ast)<0$.
\end{description}
\end{lemma}

\gs{I did not go through this proof yet.}

\begin{proof}
We will only show (i), as the proof for (ii) is similar.

Since $v_t(L,\acc)<v_t(H,\acc)$ and $v_t(L,\rej)>v_t(H,\rej)$ for any agent $t$ (see (\ref{eqn:ut_natural})), we have $v_t(L,\acc)-v_t(H,\acc)<0<v_t(L,\rej)-v_t(H,\rej)$, which further implies $$v_{t}(L,\acc)-v_{t}(L,\rej)<v_{t}(H,\acc)-v_{t}(H,\rej).$$
Since a candidate-friendly agent always prefers \acc and a candidate-unfriendly agent always prefers \rej, we have
\begin{equation}\label{eqn:monotonicity}
\begin{array}{l}
     0<v_{t_1}(L,\acc)-v_{t_1}(L,\rej)<v_{t_1}(H,\acc)-v_{t_1}(H,\rej),\mbox{ and }  \\
     v_{t_2}(L,\acc)-v_{t_2}(L,\rej)<v_{t_2}(H,\acc)-v_{t_2}(H,\rej)<0.
\end{array}
\end{equation}
By referring to (\ref{eqn:utsigma3_main}), this intuitively says that $t_1$'s utility difference $u_{t_1}(\Sigma')-u_{t_1}(\Sigma^\ast)$ is more sensitive to $P_H\lambda_H^\acc$ while $u_2$'s utility difference $u_{t_2}(\Sigma')-u_{t_2}(\Sigma^\ast)$ is more sensitive to $P_L\lambda_L^\acc$.


Suppose $u_{t_1}(\Sigma')-u_{t_1}(\Sigma^\ast)>\Delta$ as it is assumed in (i).
By (\ref{eqn:utsigma3_main}), we have
\begin{equation}\label{eqn:conflictlemma1}
     P_L(\lambda_L^\acc(\Sigma')-\lambda_L^\acc(\Sigma^\ast))(v_{t_1}(L,\acc)-v_{t_1}(L,\rej))  
     +P_H(\lambda_H^\acc(\Sigma')-\lambda_H^\acc(\Sigma^\ast))(v_{t_1}(H,\acc)-v_{t_1}(H,\rej))>\Delta.
\end{equation}
We consider two cases: $\lambda_H^\acc(\Sigma')\geq\lambda_H^\acc(\Sigma^\ast)$ and $\lambda_H^\acc(\Sigma')<\lambda_H^\acc(\Sigma^\ast)$.

The intuitions for the remaining part of this proof is as follows.
In the first case, the probability of outputting \acc (weakly) increases under world $H$. 
Since Theorem~\ref{thm:msw} tells us $\lambda_H^\acc(\Sigma^\ast)$ is already close to $1$, the utility gain for $t_1$ due to the increased probability of outputting \acc under world $H$ is insignificant, and we must still have $\lambda_L^\acc(\Sigma')-\lambda_L^\acc(\Sigma^\ast)>0$ to ensure $u_{t_1}(\Sigma')-u_{t_1}(\Sigma^\ast)>\Delta$.
However, since the probability of outputting \acc increases under both worlds, the utility for $t_2$ will decrease.
In the second case, the probability of outputting \acc decreases under world $H$, which reduces the utility for $t_1$.
To compensate this, the probability of outputting \acc under world $L$ must increase, in order to ensure $u_{t_1}(\Sigma')-u_{t_1}(\Sigma^\ast)>\Delta$.
Moreover, we must have $P_L(\lambda_L^\acc(\Sigma')-\lambda_L^\acc(\Sigma^\ast))>P_H(\lambda_H^\acc(\Sigma^\ast)-\lambda_H^\acc(\Sigma'))$ since the utility difference for agent $t_1$ is more sensitive to $P_H\lambda_H^\acc$.
However, since the utility difference for agent $t_2$ is more sensitive to $P_L\lambda_L^\acc$, this will reduce the overall utility for $t_2$.
These are formally proved below.


\paragraph{Case 1: $\lambda_H^\acc(\Sigma')\geq\lambda_H^\acc(\Sigma^\ast)$.}
By Theorem~\ref{thm:msw}, we have $\lambda_H^\acc(\Sigma^\ast)\geq1-2\exp(-2c^2\alpha_{C}T)$, which implies $\lambda_H^\acc(\Sigma')-\lambda_H^\acc(\Sigma^\ast)\leq2\exp(-2c^2\alpha_{C}T)$, which further implies
$$P_H(\lambda_H^\acc(\Sigma')-\lambda_H^\acc(\Sigma^\ast))(v_{t_1}(H,\acc)-v_{t_1}(H,\rej))\leq P_H\cdot2\exp(-2c^2\alpha_{C}T)\cdot B<\Delta.$$
Putting this into (\ref{eqn:conflictlemma1}), we have $P_L(\lambda_L^\acc(\Sigma')-\lambda_L^\acc(\Sigma^\ast))(v_{t_1}(L,\acc)-v_{t_1}(L,\rej))>0$, which implies $\lambda_L^\acc(\Sigma')>\lambda_L^\acc(\Sigma^\ast)$.
We then must have $u_{t_2}(\Sigma')-u_{t_2}(\Sigma^\ast)<0$ since we have
$$
\begin{array}{l}
     u_{t_2}(\Sigma')-u_{t_2}(\Sigma^\ast)
    =P_L(\lambda_L^\acc(\Sigma')-\lambda_L^\acc(\Sigma^\ast))(v_{t_2}(L,\acc)-v_{t_2}(L,\rej)) \\
     \qquad+P_H(\lambda_H^\acc(\Sigma')-\lambda_H^\acc(\Sigma^\ast))(v_{t_2}(H,\acc)-v_{t_2}(H,\rej))
\end{array}
$$
by (\ref{eqn:utsigma3_main}), $\lambda_H^\acc(\Sigma')\geq\lambda_H^\acc(\Sigma^\ast)$ (Case 1 assumption), $\lambda_L^\acc(\Sigma')>\lambda_L^\acc(\Sigma^\ast)$ (we have just shown), $v_{t_2}(L,\acc)-v_{t_2}(L,\rej)<0$ and $v_{t_2}(H,\acc)-v_{t_2}(H,\rej)<0$ (since $t_2$ is candidate-unfriendly).

\paragraph{Case 2: $\lambda_H^\acc(\Sigma')<\lambda_H^\acc(\Sigma^\ast)$.}
By (\ref{eqn:conflictlemma1}) and $\Delta>0$, we have
$$P_L(\lambda_L^\acc(\Sigma')-\lambda_L^\acc(\Sigma^\ast))(v_{t_1}(L,\acc)-v_{t_1}(L,\rej))>P_H(\lambda_H^\acc(\Sigma^\ast)-\lambda_H^\acc(\Sigma'))(v_{t_1}(H,\acc)-v_{t_1}(H,\rej)),$$
which, by the first inequality in (\ref{eqn:monotonicity}), further implies
$$P_L(\lambda_L^\acc(\Sigma')-\lambda_L^\acc(\Sigma^\ast))>P_H(\lambda_H^\acc(\Sigma^\ast)-\lambda_H^\acc(\Sigma'))>0.$$
By the second inequality in (\ref{eqn:monotonicity}), this implies
$$P_L(\lambda_L^\acc(\Sigma')-\lambda_L^\acc(\Sigma^\ast))(v_{t_2}(L,\acc)-v_{t_2}(L,\rej))<P_H(\lambda_H^\acc(\Sigma^\ast)-\lambda_H^\acc(\Sigma'))(v_{t_2}(H,\acc)-v_{t_2}(H,\rej)),$$
which further implies
$$u_{t_2}(\Sigma')-u_{t_2}(\Sigma^\ast)=P_L(\lambda_L^\acc(\Sigma')-\lambda_L^\acc(\Sigma^\ast))(v_{t_2}(L,\acc)-v_{t_2}(L,\rej))$$
$$\qquad+P_H(\lambda_H^\acc(\Sigma')-\lambda_H^\acc(\Sigma^\ast))(v_{t_2}(H,\acc)-v_{t_2}(H,\rej))<0.$$
The lemma concludes.
\end{proof}

\begin{corollary}\label{cor:conflict}
Suppose $\alpha_{F}<0.5$ and $\alpha_{U}<0.5$.
The set of deviating agents $D$ cannot contain both candidate-friendly and candidate-unfriendly agents.
\end{corollary}
\begin{proof}
By 3 in Definition~\ref{def:SBNE}, there must be an agent $t$ such that $u_t(\Sigma')-u_t(\Sigma^\ast)>\varepsilon\geq 2B\exp(-2c^2\alpha_{C}T)$.  Assume this agent is candidate-friendly.  Then by Lemma~\ref{lem:conflict}, for any candidate-unfriendly agent $t'$, we have $u_{t'}(\Sigma')-u_{t'}(\Sigma^\ast)<0$.  Thus no candidate-unfriendly agent can be in the deviating coalition.  

An analogous argument works if the benefiting agent is candidate-unfriendly.
\end{proof}

\subsubsection{Proof of Theorem~\ref{thm:SBNE2}}

Now we are ready to prove Theorem~\ref{thm:SBNE2}.

Suppose this is not the case.
There exists a set of deviating agents $D$ that can deviate from the truthful strategy such that all of them receive utilities that are at least their original utilities and some of them receive utilities that are $\varepsilon$ strictly higher than their original utilities.
Let $\Sigma'$ be the strategy profile after agents in $D$ deviate.

We discuss three different cases: 1) $\alpha_{F}>0.5$, 2) $\alpha_{U}>0.5$ and 3) $\alpha_{F}<0.5$ and $\alpha_{U}<0.5$.
Notice that $n$ being odd implies neither $\alpha_{F}$ nor $\alpha_{U}$ can be exactly $0.5$.

\paragraph{Case 1: $\alpha_{F}>0.5$.}  
If all agents report truthfully, \acc will be announced with probability $1$ according to Step~5 of the mechanism.
That is, $\lambda_L^\acc(\Sigma^\ast)=\lambda_H^\acc(\Sigma^\ast)=1$.
By 3 of Definition~\ref{def:SBNE}, there exists $t\in D$ such that $u_t(\Sigma')>u_t(\Sigma^\ast)+\varepsilon$.
Since $\lambda_L^\acc(\Sigma')$ and $\lambda_H^\acc(\Sigma')$ completely determine each agent's utility, we must have either $\lambda_L^\acc(\Sigma')\neq\lambda_L^\acc(\Sigma^\ast)$ or $\lambda_H^\acc(\Sigma')\neq\lambda_H^\acc(\Sigma^\ast)$.
This means either $\lambda_L^\acc(\Sigma')<1$ or $\lambda_H^\acc(\Sigma')<1$.

Since a candidate-friendly agent's utility is maximized when both $\lambda_L^\acc$ and $\lambda_H^\acc$ are $1$, a candidate-friendly agent's utility will decrease if the strategy profile is switched from $\Sigma^\ast$ to $\Sigma'$.
By 2 of Definition~\ref{def:SBNE}, $D$ does not contain any candidate-friendly agent.
However, if this is the case, there are still more than half of the agents that will report $F$ (as $\alpha_{F}>0.5$), and \acc will always be announced by Step~5 of the mechanism.
We conclude that $\lambda_L^\acc(\Sigma')=\lambda_H^\acc(\Sigma')=1$, which contradicts what we have concluded in the previous paragraph.

\paragraph{Case 2: $\alpha_{U}>0.5$.}
The analysis is similar to the previous case.

\paragraph{Case 3: $\alpha_{F}<0.5$ and $\alpha_{U}<0.5$.}

We consider two sub-cases: $I(\Sigma') < (2B+2)\exp(-2c^2\alpha_{C}T)$ and $I(\Sigma') \geq (2B+2)\exp(-2c^2\alpha_{C}T)$

Firstly, we consider $I(\Sigma') < (2B+2)\exp(-2c^2\alpha_{C}T)$.

\begin{claim}\label{claim:sameasbefore}
If $\alpha_{F}<0.5$, $\alpha_{U}<0.5$ and  $I(\Sigma') < (2B+2)\exp(-2c^2\alpha_{C}T)$, then $u_t(\Sigma')-u_t(\Sigma)<\varepsilon$ for every agent $t$.
\end{claim}
The ideas behind this proof is that the outcome for $\Sigma'$ is too close to the outcome of the truthful strategy profile $\Sigma^\ast$, so no agent can get significantly more benefit.
\begin{proof}  
The proof of this claim shows that none of the three types of agents can benefit by $\varepsilon$ because nothing is substantially different from when agents play truthfully.  

By the inequality $I(\Sigma')=P_L\lambda_L^\acc(\Sigma')+P_H\lambda_H^\rej(\Sigma') < (2B+2)\exp(-2c^2\alpha_{C}T)$,
we have 
$$0\leq\lambda_H^\rej(\Sigma')\leq\frac{(2B+2)\exp(-2c^2\alpha_{C}T)}{P_H}$$ 
and 
$$0\leq\lambda_L^\acc(\Sigma')\leq\frac{(2B+2)\exp(-2c^2\alpha_{C}T)}{P_L}.$$
Since $\lambda_L^\acc(\Sigma^\ast)\leq2\exp(-2c^2\alpha_CT)$ and $\lambda_H^\rej(\Sigma^\ast)\leq2\exp(-2c^2\alpha_CT)$ by Theorem~\ref{thm:msw}, we have
$$\lambda_L^\acc(\Sigma^\ast)-\lambda_L^\acc(\Sigma')\leq \lambda_L^\acc(\Sigma^\ast)\leq2\exp(-2c^2\alpha_CT),$$
$$\lambda_L^\acc(\Sigma')-\lambda_L^\acc(\Sigma^\ast)\leq \lambda_L^\acc(\Sigma')\leq \frac{(2B+2)\exp(-2c^2\alpha_{C}T)}{P_L},$$
$$\lambda_H^\acc(\Sigma^\ast)-\lambda_H^\acc(\Sigma')\leq 1-\left(1-\lambda_H^\rej(\Sigma')\right)\leq \frac{(2B+2)\exp(-2c^2\alpha_{C}T)}{P_H},$$
and
$$\lambda_H^\acc(\Sigma')-\lambda_H^\acc(\Sigma^\ast)\leq 1-\left(1-\lambda_H^\rej(\Sigma^\ast)\right)\leq 2\exp(-2c^2\alpha_CT).$$
Now, we substitute the four inequalities into the following equation implied by (\ref{eqn:utsigma3_main}):
$$
\begin{array}{l}
     u_t(\Sigma')-u_t(\Sigma^\ast)=P_L(\lambda_L^\acc(\Sigma')-\lambda_L^\acc(\Sigma^\ast))(v_t(L,\acc)-v_t(L,\rej))  \\
     \qquad+P_H(\lambda_H^\acc(\Sigma')-\lambda_H^\acc(\Sigma^\ast))(v_t(H,\acc)-v_t(H,\rej)). 
\end{array}
$$
For any candidate-friendly agent, we have $v_t(L,\acc)>v_t(L,\rej)$ and $v_t(H,\acc)>v_t(H,\rej)$, which yields
\begin{align*}
    &u_t(\Sigma')-u_t(\Sigma^\ast)\\
    =&P_L(\lambda_L^\acc(\Sigma')-\lambda_L^\acc(\Sigma^\ast))(v_t(L,\acc)-v_t(L,\rej))+P_H(\lambda_H^\acc(\Sigma')-\lambda_H^\acc(\Sigma^\ast))(v_t(H,\acc)-v_t(H,\rej))\\
    \leq& P_L\cdot\frac{(2B+2)\exp(-2c^2\alpha_{C}T)}{P_L}\cdot B+P_H\cdot2\exp(-2c^2\alpha_{C}T))\cdot B\\
    <&(2B^2+4B)\exp(-2c^2\alpha_{C}T)=\varepsilon.
\end{align*}
For any contingent agent, we have $v_t(L,\rej)>v_t(L,\acc)$ and $v_t(H,\acc)>v_t(H,\rej)$, which yields
\begin{align*}
    &u_t(\Sigma')-u_t(\Sigma^\ast)\\
    =&P_L(\lambda_L^\acc(\Sigma^\ast)-\lambda_L^\acc(\Sigma'))(v_t(L,\rej)-v_t(L,\acc))+P_H(\lambda_H^\acc(\Sigma')-\lambda_H^\acc(\Sigma^\ast))(v_t(H,\acc)-v_t(H,\rej))\\
    \leq& P_L\cdot2\exp(-2c^2\alpha_{C}T)\cdot B+P_H\cdot2\exp(-2c^2\alpha_{C}T)\cdot B\\
    =&2B\exp(-2c^2\alpha_{C}T)<\varepsilon.
\end{align*}
For any candidate-unfriendly agent, we have $v_t(L,\rej)>v_t(L,\acc)$ and $v_t(H,\rej)>v_t(H,\acc)$, which yields
\begin{align*}
    &u_t(\Sigma')-u_t(\Sigma^\ast)\\
    =&P_L(\lambda_L^\acc(\Sigma^\ast)-\lambda_L^\acc(\Sigma'))(v_t(L,\rej)-v_t(L,\acc))+P_H(\lambda_H^\acc(\Sigma^\ast)-\lambda_H^\acc(\Sigma'))(v_t(H,\rej)-v_t(H,\acc))\\
    \leq& P_L\cdot2\exp(-2c^2\alpha_{C}T)\cdot B+P_H\cdot\frac{(2B+2)\exp(-2c^2\alpha_{C}T)}{P_H}\cdot B\\
    <&(2B^2+4B)\exp(-2c^2\alpha_{C}T)=\varepsilon.
\end{align*}
We conclude none of the agents has a utility gain of at least $\varepsilon$, which contradict 3 of Definition~\ref{def:SBNE}.
\end{proof}

Claim~\ref{claim:sameasbefore} implies that, in the first case $I(\Sigma') <  (2B+2)\exp(-2c^2\alpha_{C}T)$, there does not exist a deviating coalition $D$ where an agent in $D$ can receive an utility gain of at least $\varepsilon$, which contradicts to our assumption about $D$ at the beginning. 

Next, we consider the second case $I(\Sigma') \geq  (2B+2)\exp(-2c^2\alpha_{C}T)$.  
This case is more complicated. 
We first show that no contingent agent in $\Sigma'$ can do as well as in the truthful profile $\Sigma^\ast$.

\begin{claim}\label{claim:contingentunhappy}
If $\alpha_{F}<0.5$, $\alpha_{U}<0.5$ and $I(\Sigma') \geq (2B+2)\exp(-2c^2\alpha_{C}T)$, then $u_t(\Sigma')-u_t(\Sigma^\ast)<0$ for every contingent agent $t$.
\end{claim}
The ideas behind this proof is that those contingent agents already receive almost optimal utilities in $\Sigma^\ast$; therefore, if the error rate of the strategy $\Sigma'$ is high enough, the utilities of the contingent agents will decrease.
\begin{proof}
By Theorem~\ref{thm:msw}, we have
\begin{align}\label{eqn:Sigmastar>sigmaprime}
    P_L\lambda_L^\rej(\Sigma^\ast)+P_H\lambda_H^\acc(\Sigma^\ast) &\geq P_L(1-2\exp(-2c^2\alpha_{C}T))+P_H(1-2\exp(-2c^2\alpha_{C}T))\nonumber\\
    &=1-2\exp(-2c^2\alpha_{C}T)\nonumber\\
    &=2B\exp(-2c^2\alpha_{C}T)+1-(2B+2)\exp(-2c^2\alpha_{C}T)\nonumber\\
    &\geq2B\exp(-2c^2\alpha_{C}T)+1-I(\Sigma')\nonumber\\
    &=2B\exp(-2c^2\alpha_{C}T)+(P_L-P_L\lambda_L^\acc(\Sigma'))+(P_H-P_H\lambda_H^\rej(\Sigma'))\nonumber\\
    &= P_L\lambda_L^\rej(\Sigma')+P_H\lambda_H^\acc(\Sigma')+2B\exp(-2c^2\alpha_{C}T)
\end{align}
By (\ref{eqn:utsigma2_main}), we have
\begin{equation}\label{eqn:differenceBetweenSigma}
\begin{array}{l}
     u_t(\Sigma')-u_t(\Sigma^\ast)
  =  P_L(v_t(L,\rej)-v_t(L,\acc))(\lambda_L^\rej(\Sigma')-\lambda_L^\rej(\Sigma^\ast))  \\
     \qquad\qquad+P_H(v_t(H,\acc)-v_t(H,\rej))(\lambda_H^\acc(\Sigma')-\lambda_H^\acc(\Sigma^\ast))
\end{array}
\end{equation}
We will show $u_t(\Sigma')-u_t(\Sigma^\ast)<0$ for an arbitrary contingent agent $t$.
Recall that $v_t(L,\rej)-v_t(L,\acc)>0$ and $v_t(H,\acc)-v_t(H,\rej)>0$.
We consider three cases:
\begin{itemize}
    \item If $\lambda_H^\acc(\Sigma')\leq\lambda_H^\acc(\Sigma^\ast)$ and $\lambda_L^\rej(\Sigma')\leq\lambda_L^\rej(\Sigma^\ast)$, then one of these two inequalities must be strict by (\ref{eqn:Sigmastar>sigmaprime}). Equation (\ref{eqn:differenceBetweenSigma}) then implies $u_t(\Sigma')-u_t(\Sigma^\ast)<0$.
    \item If $\lambda_H^\acc(\Sigma')>\lambda_H^\acc(\Sigma^\ast)$, then we have  $P_L\lambda_L^\rej(\Sigma^\ast)-P_L\lambda_L^\rej(\Sigma')>2B\exp(-2c^2\alpha_{C}T)$ by (\ref{eqn:Sigmastar>sigmaprime}).
    Since $\lambda_H^\acc(\Sigma')\leq1$ and $\lambda_H^\acc(\Sigma^\ast)\geq1-2\exp(-2c^2\alpha_{C}T)$, we have $\lambda_H^\acc(\Sigma')-\lambda_H^\acc(\Sigma^\ast)\leq2\exp(-2c^2\alpha_{C}T)$.
    We also have $v_t(H,\acc)-v_t(H,\rej)\leq B$ and $v_t(L,\rej)-v_t(L,\acc)\geq1$ (recall that  $v_t(L,\acc)$, $v_t(L,\rej)$, $v_t(H,\rej)$ and $v_t(H,\acc)$ are integers bounded by $B$).
    Putting those into (\ref{eqn:differenceBetweenSigma}), we have
    $$u_t(\Sigma')-u_t(\Sigma^\ast)<1\cdot(-2B\exp(-2c^2\alpha_{C}T))+P_H\cdot B\cdot2\exp(-2c^2\alpha_{C}T)<0.$$
    \item If $\lambda_L^\rej(\Sigma')>\lambda_L^\rej(\Sigma^\ast)$, then we have $P_H\lambda_H^\acc(\Sigma^\ast)-P_H\lambda_H^\acc(\Sigma')>2B\exp(-2c^2\alpha_{C}T)$ by (\ref{eqn:Sigmastar>sigmaprime}).
    Similar to the second case, we have $\lambda_L^\rej(\Sigma')-\lambda_L^\rej(\Sigma^\ast)\leq2\exp(-2c^2\alpha_{C}T)$,  $v_t(H,\acc)-v_t(H,\rej)\geq1$ and $v_t(L,\rej)-v_t(L,\acc)\leq B$.
    Putting those into (\ref{eqn:differenceBetweenSigma}), we have
    $$u_t(\Sigma')-u_t(\Sigma^\ast)<P_L\cdot B\cdot2\exp(-2c^2\alpha_{C}T)+1\cdot(-2B\exp(-2c^2\alpha_{C}T))<0.$$
\end{itemize}
Putting these three cases together, we have $u_t(\Sigma')-u_t(\Sigma^\ast)<0$ for an arbitrary contingent agent $t$. 
\end{proof}

Therefore, $D$ cannot contain any contingent agents by 2 of Definition~\ref{def:SBNE}.
Corollary~\ref{cor:conflict} says that $D$ cannot simultaneously contain an candidate-friendly agent and a candidate-unfriendly agent.  Thus any $D$ must contain either only candidate-friendly agents or only candidate-unfriendly agents.

The following claim states that neither type of predetermined agents alone are powerful enough to change the outcome to their favor.  This concludes the proof as we have shown there is no deviating coalition. 

\gs{ the below claim was completely wrong.  I changed it and the proof.}
\bt{I see. There was a typo, the inequality was mistyped to equality.}
\bt{I rewrite the lemma and the proof below. I think the original proof is for the version of the mechanism where we have converted the predictions to the extreme points for type $F$ and $U$ agents.}

\begin{claim}  \label{claim:nochange}
Suppose $\alpha_{F}<0.5$ and $\alpha_{U}<0.5$.  If $D$ contains only candidate-friendly agents, then $\lambda_L^\acc(\Sigma') \leq  \lambda_L^\acc(\Sigma^*)$ and  $\lambda_H^\acc(\Sigma') \leq  \lambda_H^\acc(\Sigma^*)$.  If $D$ contains only candidate-unfriendly agents, then $\lambda_L^\rej(\Sigma') \leq  \lambda_L^\rej(\Sigma^*)$ and  $\lambda_H^\rej(\Sigma') \leq  \lambda_H^\rej(\Sigma^*)$.
\end{claim}
\begin{proof}
Consider candidate-friendly agents without loss of generality.
Since contingent agents and candidate-unfriendly agents, which constitute more than half of the population, are truth-telling, those candidate-friendly agents cannot make the mechanism announce \acc at Step~5, since they cannot make more than half of agents report type $F$.
To maximize the probability that the mechanism announce \acc at Step~6, those candidate-friendly agents would like to maximize the fraction of agents reporting signal $h$ and minimize the median $\bar\delta$.
However, the mechanism's conversion of signals (Step~2) and predictions (Step~3) already does these for candidate-friendly agents.
\end{proof}

Thus, we have proved that, for the second case $I(\Sigma') \geq  (2B+2)\exp(-2c^2\alpha_{C}T)$, no such deviating set $D$ exists, which contradicts to our assumption for the existence of $D$ at the beginning.

\subsection{An Alternative Mechanism}
\label{sect:alternative_mechanism}
As we have remarked right below Mechanism~\ref{mechanism:M=N=2}, we present an alternative mechanism that achieves the same theoretical properties.
The mechanism is shown in Mechanism~\ref{mechanism:alternative}.

\begin{algorithm}
\caption{The Wisdom-of-the-Crowd-Voting Mechanism (an alternative)}
\label{mechanism:alternative}
\begin{algorithmic}[1]
    \STATE Each agent $t$ reports to the mechanism the signal (s)he receives (either $\ell$ or $h$), denoted by $\bar{s}_i\in\{\ell,h\}$, his/her type ($F$, $U$ or $C$), his/her posterior belief of the fraction of agents who will report signal $h$, denoted by $\bar{\delta}_t$.
    \STATE If agent $t$ reports type $F$, his reported signal will be automatically treated as $\bar{s}_t=h$; 
    if agent $t$ reports type $U$, his reported signal will be automatically treated as $\bar{s}_t=\ell$. 
    \emph{The prediction $\bar{\delta}_t$ in the previous step should be made with this treatment being considered, and the mechanism makes this clear to the agents.}
    \STATE Compute the \emph{median} of the reported $\bar{\delta}_t$, denoted by $\bar\delta$.
    \STATE If more than half of the agents report type $F$, announce \acc being the winning alternative; if more than half of the agents reports type $U$, announce \rej being the winning alternative.
    \STATE If the number of agents reporting $\bar{s}_t=h$ is more than the median $\bar\delta$, announce \acc being the winning alternative; otherwise, announce \rej being the winning alternative.
\end{algorithmic}
\end{algorithm}

The only difference between this mechanism and Mechanism~\ref{mechanism:M=N=2} is that we ask all the agents to report their predictions without any changes or treatments afterwards.

Correspondingly, the questionnaire becomes the followings.

\begin{enumerate}
    \item Choose one of the followings:
    \begin{enumerate}
        \item I definitely want to accept this candidate.
        \item I definitely want to reject this candidate.
        \item After talking to the candidate, I am more inclined to accept him/her than before.
        \item After talking to the candidate, I am more inclined to reject him/her than before.
    \end{enumerate}
    \item What percentage of the faculty members do you believe will choose (a) or (c) in the first question?
\end{enumerate}

Theorem~\ref{thm:msw} still holds for Mechanism~\ref{mechanism:alternative}.
If $\alpha_F>0.5$ or $\alpha_U>0.5$, the mechanism outputs the majority wish (\acc or \rej respectively) at Step~4 with probability $1$ as before.
If $\alpha_F,\alpha_U<0.5$, we still have $\bar\delta\in[\alpha_CT_{h\ell}+\alpha_F,\alpha_CT_{hh}+\alpha_F]$.
This is actually easier to see: agents' predictions are now either $\alpha_CT_{h\ell}+\alpha_F$ or $\alpha_CT_{hh}+\alpha_F$.
The remaining part of the proof is the same as before.

Theorem~\ref{thm:SBNE2} still holds for Mechanism~\ref{mechanism:alternative}.
In fact, all parts of the proof are the same as before, except for Claim~\ref{claim:nochange} where we can only prove a weaker statement, which is, nevertheless, sufficient to show Theorem~\ref{thm:SBNE2}.

\begin{claim}  \label{claim:nochange--}
Suppose $\alpha_{F}<0.5$ and $\alpha_{U}<0.5$.  If $D$ contains only candidate-friendly agents, then $\lambda_L^\acc(\Sigma') \leq  \lambda_L^\acc(\Sigma^*)+2\exp(-2c^2\alpha_CT)$ and  $\lambda_H^\acc(\Sigma') \leq  \lambda_H^\acc(\Sigma^*)+2\exp(-2c^2\alpha_CT)$.  If $D$ contains only candidate-unfriendly agents, then $\lambda_L^\rej(\Sigma') \leq  \lambda_L^\rej(\Sigma^*)+2\exp(-2c^2\alpha_CT)$ and  $\lambda_H^\rej(\Sigma') \leq  \lambda_H^\rej(\Sigma^*)+2\exp(-2c^2\alpha_CT)$.
\end{claim}
\begin{proof}
We focus on the case that $D$ contains only candidate-friendly agents.  The candidate-unfriendly case is analogous.   

First of all, since there are strictly less than half of the agents reporting type $F$ (those type $U$ and type $C$ agents, which contribute more than half of the population, report their types truthfully), those candidate-friendly agents in $D$ cannot make the mechanism output \acc at Step~5 of the mechanism.
Therefore, they can only attempt to make the mechanism output \acc at Step~6 with a higher probability.

Suppose the actual world is $H$.
We need to prove $\lambda_H^\acc(\Sigma') \leq  \lambda_H^\acc(\Sigma^*)+2\exp(-2c^2\alpha_CT)$.
This is trivial: Theorem~\ref{thm:msw} implies $\lambda_H^\acc(\Sigma^*)\geq1-2\exp(-2c^2\alpha_CT)$, so the right-hand side of the inequality is at least $1$, making the inequality always hold.

Suppose the actual world is $L$.
We need to prove $\lambda_L^\acc(\Sigma') \leq  \lambda_L^\acc(\Sigma^*)+2\exp(-2c^2\alpha_CT)$.
It suffices to show $\lambda_L^\acc(\Sigma') \leq  2\exp(-2c^2\alpha_CT)$, which is equivalent to
\begin{equation}\label{eqn:nochange_finalgoal}
\lambda_L^\rej(\Sigma')\geq1-2\exp(-2c^2\alpha_CT).
\end{equation}
Supposing $\Sigma'$ is played, we will prove the following two observations:
\begin{enumerate}
    \item With probability at least $1-2\exp(-2c^2\alpha_CT)$, the fraction of agents reporting signal $h$ is at most $\alpha_C(P_{hL}+c)+\alpha_F$;
    \item The median of the prediction $\bar\delta$ falls into the interval $[\alpha_CT_{h\ell}+\alpha_F,\alpha_CT_{hh}+\alpha_F]$ (with probability $1$).
\end{enumerate}
To show the first observation, all the candidate-unfriendly agents will report signal $\ell$ (after the conversion in Step~2).
For the contingent agents, each of them receives signal $h$ with probability $P_{hL}$.
By a Chernoff bound, with probability at least $1-2\exp(-2c^2\alpha_CT)$, the fraction of the contingent agents receiving $h$ is at most $P_{hL}+c$.
Even if all the candidate-friendly agents report $h$, the overall fraction of agents reporting $h$ is at most $\alpha_C(P_{hL}+c)+\alpha_F$ with probability at least $1-2\exp(-2c^2\alpha_CT)$.

To show the second observation, the prediction $\bar{\delta}_t$ that a truthful agent will report is either $\alpha_CT_{h\ell}+\alpha_F$ (if (s)he receive signal $\ell$) or $\alpha_CT_{hh}+\alpha_F$ (if (s)he receive signal $h$).
Since there are more than half of truth-telling agents, the median $\bar\delta$ is always within the interval $[\alpha_CT_{h\ell}+\alpha_F,\alpha_CT_{hh}+\alpha_F]$.

Finally, by noticing $\alpha_C(P_{hL}+c)+\alpha_F<\alpha_CT_{h\ell}+\alpha_F$ (implied by Theorem~\ref{thm:keyInequality} and (\ref{eqn:C})), the two observations imply that the fraction of agents reporting signal $h$ is strictly less than $\bar\delta$ with probability at least $1-2\exp(-2c^2\alpha_CT)$, which implies (\ref{eqn:nochange_finalgoal}) by our design in Step~6 of the mechanism.
\end{proof}

To conclude Theorem~\ref{thm:SBNE2}, for each agent $t$ in $D$ that contains only candidate-friendly agents, we have
\begin{align*}
    & u_t(\Sigma')-u_t(\Sigma^\ast)\\
    =&P_L(\lambda_L^\acc(\Sigma')-\lambda_L^\acc(\Sigma^\ast))(v_t(L,\acc)-v_t(L,\rej))+P_H(\lambda_H^\acc(\Sigma')-\lambda_H^\acc(\Sigma^\ast))(v_t(H,\acc)-v_t(H,\rej))\tag{by (\ref{eqn:utsigma3_main})}\\
    \leq &P_L\cdot 2\exp(-2c^2\alpha_CT) \cdot B+P_H\cdot 2\exp(-2c^2\alpha_CT)\cdot B\tag{by Claim~\ref{claim:nochange--}}\\
    =& 2B\exp(-2c^2\alpha_CT)<\varepsilon.
\end{align*}
Thus, no agent in $D$ satisfies 3 in Definition~\ref{def:SBNE}.
A similar analysis holds for the case where $D$ contains only candidate-unfriendly agents.

\subsubsection{Comparison of the Two Mechanisms}
Both Mechanism~\ref{mechanism:M=N=2} and Mechanism~\ref{mechanism:alternative} achieve the same set of theoretical properties.

The advantage for Mechanism~\ref{mechanism:M=N=2} is that it is ``slightly more truthful'' in the sense that Claim~\ref{claim:nochange} is stronger than Claim~\ref{claim:nochange--}.
In fact, under Mechanism~\ref{mechanism:M=N=2}, we have seen that truth-telling is a dominant strategy for both candidate-friendly and candidate-unfriendly agents, while this nice property is lost in Mechanism~\ref{mechanism:alternative}.
Under Mechanism~\ref{mechanism:alternative}, the dominant strategy for a candidate-friendly agent (candidate-unfriendly agent resp.) is to report prediction $0$ ($1$ resp.), which is no longer a truthful strategy.
Nevertheless, we have seen that the truthful strategy is good enough so that a deviation to the dominant strategy does not provide a utility gain of at least $\varepsilon$.

Mechanism~\ref{mechanism:alternative} wins by a little bit for its simplicity and symmetry.
It is easier to explain Mechanism~\ref{mechanism:alternative} to the users in practice.
Notice that Mechanism~\ref{mechanism:M=N=2} essentially converts the prediction reported from each candidate-friendly agent (candidate-unfriendly agent resp.) to $0$ ($1$ resp.).
Converting the predictions may seem to be less natural than converting the signals for users.
Especially, for those users who are not familiar to the idea of ``surprisingly popular'', they may not be able to realize that converting their predictions to the opposite extreme is helpful for them, and they may be more skeptical of Mechanism~\ref{mechanism:M=N=2} due to this.
In addition, Mechanism~\ref{mechanism:alternative} treats the reported predictions symmetrically, which may be more acceptable to the users in practical implementations.

\gs{I added this.  Note sure this is the best place for it.}
\bt{I think the paragraph below may be out of place. However, I do not find a better place to put it.}
Another related question is how exactly to phrase the ballot in practice.  In both Mechanism~\ref{mechanism:M=N=2} and \ref{mechanism:alternative}, we mimic the questions of \citet{prelec2017solution} and ask for a forecast.  However, it may be preferable in practice to ask, as in Mechanism~\ref{mechanism:general}, for a fractional threshold of (a) and (c) responses above which the agents would prefer to accept.  While the outcomes would be mathematically equivalent, one or the other or a third alternative might work better in practice.  This is, however, beyond the scope of this paper.

\section{Unknown/Partially Known Distribution of Agent Types}
\label{sect:unknowndistribution}
As we mentioned in Section~\ref{sect:prelim_main}, we assume the distribution of agent types, $\alpha_{F},\alpha_{U},\alpha_{C}$, is a common knowledge among the agents.
This assumption is natural by its own in many scenarios including our candidate hiring example (if a theory candidate is applying at a computer science department, those theory faculty members are more inclined to accept the candidate than the faculty members in AI, software, hardware; moreover, the numbers of the theory, AI, software, hardware faculty members are public information).
In this section, we will see that this assumption is also necessary for the existence of a mechanism that satisfies Theorem~\ref{thm:msw} and Theorem~\ref{thm:SBNE2}.

Before describing our impossibility result, we first formally define the model with unknown agent types.
Let $\Delta_{3}=\{(x_1,x_2,x_3)\mid \forall i:x_i \in [0,1] ;x_1+x_2+x_3=1\}$.
The distribution of types, $(\alpha_{F},\alpha_{U},\alpha_{C})$, is then an element of $\Delta_{3}$.
To model an unknown/partially known distribution of agent types, let $\mathcal{D}_{\Delta_{3}}$ be a distribution over $\Delta_{3}$ where each agent believes the distribution of the agent types, $(\alpha_{F},\alpha_{U},\alpha_{C})$, is drawn from $\mathcal{D}_{\Delta_{3}}$. 

\gs{added this}
Note that while the fraction of types is not known, the prior over the worlds, $P_L, P_H$, and the signal structures conditioned on types, $P_{\ell L}, P_{\ell H}, P_{h L}, P_{h H} $ are still common knowledge.


Next, we describe a natural property that is shared by most social choice mechanism, including the one in this paper.

\begin{definition}
A mechanism is \emph{anonymous} if it always outputs the same alternative for any two collections of reports $\mathbf{r}^{(1)}=(r_1^{(1)},\ldots,r_T^{(1)})\in\calR^T, \mathbf{r}^{(2)}=(r_1^{(2)},\ldots,r_T^{(2)})\in\calR^T$ such that $\mathbf{r}^{(1)}$ is a permutation of $\mathbf{r}^{(2)}$.
\end{definition}

In other words, an anonymous mechanism cannot decide the output alternative based on agents' identities.

We have the following strong impossibility result.

\begin{theorem}\label{thm:impossibility}
Under the setting with an unknown distribution of agent types, there exists a constant $\tau>0$ such that no anonymous mechanism always outputs the alternative favored by more than half of the agents with probability more than $1-\tau$ in any $\tau$-strong symmetric Bayes Nash equilibrium. 
\end{theorem}

Since the truthful strategy profile is symmetric, we have the following corollary about the impossibility of a truthful mechanism.

\begin{corollary}
Under the setting with unknown distribution of agent types defined above, there exists a constant $\tau>0$ such that no anonymous mechanism satisfies both of the followings:
\begin{itemize}
    \item the mechanism outputs the alternative favored by more than half of the agents with probability more than $1-\tau$;
    \item under the mechanism, the truthful strategy profile is a $\tau$-strong Bayes Nash equilibrium. 
\end{itemize} 
\end{corollary}

\subsection{Proof of Theorem~\ref{thm:impossibility}}
\label{append:proof4}
Consider an anonymous mechanism and an arbitrary symmetric strategy profile $\Sigma$.
Let $\beta_{F,\ell}$ be the fraction of candidate-friendly agents that receive signal $\ell$.
Let $\beta_{F,h},\beta_{U,\ell},\beta_{U,h},\beta_{C,\ell}$ and $\beta_{C,h}$ have similar meanings.
The mechanism can only see how many different reports there are, and how many agents report each of them;
in particular, the mechanism cannot see who reports which.
The following proposition follows immediately from the above remarks.

\begin{proposition}\label{prop:anonymous}
Fix a symmetric strategy profile $\Sigma$.
If a mechanism is anonymous, then the values $\{\alpha_F,\alpha_U,\alpha_C\}\cup\{\beta_{F,\ell},\beta_{F,h},\beta_{U,\ell},\beta_{U,h},\beta_{C,\ell},\beta_{C,h}\}$ completely determine the output of the mechanism.
\end{proposition}

We also need the following technical lemma.
\begin{lemma}\label{lem:tvd}
The total variation distance between the two binomial distributions $\Bin(T, 1/6)$ and $\Bin(T/3, 1/2)$ is less than $0.123$ for sufficiently large $T$.
\end{lemma}
\begin{proof}
By Central Limit Theorem, the total variation distance between $\Bin(T, 1/6)$ and $\Bin(T/3, 1/2)$ is at most the total variation distance between $\mathcal{N}(T/6, 5T/36)$ and $\mathcal{N}(T/6, T/12)$ plus $o(1)$, which, by shifting the mean of the Gaussian distribution, is the total variation distance between $\mathcal{N}(0, 5T/36)$ and $\mathcal{N}(0, T/12)$ plus $o(1)$.

Let $f(x)$ and $g(x)$ be the probability density function for $\mathcal{N}(0, 5T/36)$ and $\mathcal{N}(0, T/12)$ respectively.
To calculate the total variation distance, firstly, straightforward calculations reveal that $f(x)<g(x)$ on $\left(-\sqrt{\frac{5}{24}\ln\frac{5}{3}T},\sqrt{\frac{5}{24}\ln\frac{5}{3}T}\right)$ and $f(x)>g(x)$ on $\left(-\infty,-\sqrt{\frac{5}{24}\ln\frac{5}{3}T}\right)\cup\left(\sqrt{\frac{5}{24}\ln\frac{5}{3}T},\infty\right)$.
Therefore, the total variation distance between $\mathcal{N}(0, 5T/36)$ and $\mathcal{N}(0, T/12)$ is
\begin{align*}
    \int_{-\sqrt{\frac{5}{24}\ln\frac{5}{3}T}}^{\sqrt{\frac{5}{24}\ln\frac{5}{3}T}} g(x)-f(x)dx&=\int_{-\sqrt{\frac{5}{24}\ln\frac{5}{3}T}}^{\sqrt{\frac{5}{24}\ln\frac{5}{3}T}}\frac1{\sqrt{2\pi \frac{T}{12}}}e^{-\frac12\frac{x^2}{T/12}}-\frac1{\sqrt{2\pi \frac{5T}{36}}}e^{-\frac12\frac{x^2}{5T/36}}dx\\
    &=\int_{-\sqrt{\frac{5}{24}\ln\frac{5}{3}}}^{\sqrt{\frac{5}{24}\ln\frac{5}{3}}}\frac1{\sqrt{2\pi \frac{1}{12}}}e^{-\frac12\frac{y^2}{1/12}}-\frac1{\sqrt{2\pi \frac{5}{36}}}e^{-\frac12\frac{y^2}{5/36}}dy\tag{where $y=x/\sqrt{T}$}\\
    &<0.12295.\tag{Calculated by computer}
\end{align*}
Thus, the total variation distance between $\Bin(T, 1/6)$ and $\mathcal{N}(T/6, 5T/36)$ is at most $0.12295+o(1)$, which implies the lemma.
\end{proof}

Now we are ready to present the proof of Theorem~\ref{thm:impossibility}.
\begin{proof}[Proof of Theorem~\ref{thm:impossibility}]
We consider the following instance. 

The prior distribution of the two worlds (world $L$ and world $H$) is given by $P_L=0.98$ and $P_H=0.02$.
The probability distribution of the two signals under each of the two worlds is given by $P_{\ell L}=1$, $P_{\ell H}=0$, $P_{hL}=5/6$, $P_{hH}=1/6$.
For each candidate-friendly agent $t$, we have $v_t(H,\acc)=3,v_t(H,\rej)=0,v_t(L,\acc)=2$ and $v_t(L,\rej)=1$.
For each contingent agent $t$ we have $v_t(H,\acc)=3,v_t(H,\rej)=0,v_t(L,\acc)=1$ and $v_t(L,\rej)=2$.
For each candidate-unfriendly agent $t$, we have $v_t(H,\acc)=1,v_t(H,\rej)=2,v_t(L,\acc)=0$ and $v_t(L,\rej)=3$.
Lastly, $\mathcal{D}_{\Delta_3}$ is defined as follows: with probability $1/2$ we are in setting $X$ and the fractions of agents with types $F,C,U$ are $\alpha_{F}^{(1)}=1/3,\alpha_{C}^{(1)}=2/3,\alpha_{U}^{(1)}=0$ respectively; with probability $1/2$ we are in setting $Y$, the fractions of agents with types $F,C,U$ are $\alpha_{F}^{(2)}=0,\alpha_{C}^{(2)}=1,\alpha_{U}^{(2)}=0$ respectively.
This finishes the description of the instance.
Note there are 2 worlds and 2 settings yielding 4 possible environments which we label $LX$, $LY$, $HX$, and $HY$.

Let $\tau = 0.001$. 
Suppose there exists a mechanism $\mathcal{M}$ that outputs the majority wish with probability at least $1 - \tau$ in a symmetric strategy profile $\Sigma$.   
We will show that $\Sigma$ cannot be a $\tau$-strong Bayes Nash equilibrium.    

Notice that the contingent agents are the majority in all the four settings.
In both environments $HX$ and $HY$, each of which happens with probability $1\%$, the majority wish is always to accept.
Thus, the mechanism must accept with probability at least $99\%$ in each environment for otherwise it will be far from being achieving $1-\tau$ accuracy.  
Similarly, in both environments $LX$ and $LY$,  each of which occurs with probability $49\%$, the mechanism must accept with probability at most $1\%$ for otherwise it will be far from being achieving $1-\tau$ accuracy. 

To show that $\Sigma$ cannot be a $\tau$-strong Bayes Nash equilibrium, consider that the set of deviating agents are all the candidate-friendly agents. 
Those candidate-friendly agents pretend they are contingent agents such that signal $\ell$ is received with probability $1/2$ and signal $h$ with probability $1/2$.
These candidate-friendly agents will follow the strategy of the real contingent agents according to $\Sigma$.
Let $\Sigma'$ be the resultant strategy profile.
We aim to show that those deviating candidate-friendly agents can increase their utilities significantly in $\Sigma'$.

For an intuitive argument, in $LX$, the mechanism sees that all the agents are contingent, and the fraction of agents receiving signal $h$ follows distribution $\Bin(T/3,1/2)$; in $HY$, the mechanism also sees that all the agents are contingent, and the fraction of agents receiving signal $h$ follows distribution $\Bin(T,1/6)$.
Lemma~\ref{lem:tvd} implies that the mechanism cannot distinguish between environments $LX$ and $HY$ with probability more than $87.7\%$.
Before deviating, the mechanism will output $\acc$ with probability at most $1\%$ in $LX$; after deviating, the mechanism will output $\acc$ with probability at least $99\%\cdot 87.7\%>1\%$ in $LX$ by confusing $LX$ with $HY$.
The candidate-friendly agents will benefit from deviating.

To make the arguments in the previous paragraph more rigorous, let $\mathcal{M}(X)=1$ if the mechanism outputs \acc when all the agents are contingent, agents play according to $\Sigma$, and the number of agents receiving signal $h$ is $X$.
Let $\mathcal{M}(X)=0$ if the output is \rej under the same circumstance.
This definition is well-defined due to Proposition~\ref{prop:anonymous}.

In $HY$, the mechanism outputs \acc with probability $\E_{X\sim\Bin(T,1/6)}[\mathcal{M}(X)]$.
In $LX$, when the candidate-friendly agents deviate to $\Sigma'$, the mechanism outputs \acc with probability $\E_{X\sim\Bin(T/3,1/2)}[\mathcal{M}(X)]$.
Lemma~\ref{lem:tvd} implies\bt{Do we need to cite any theorem for the inequality below?} $$\left|\E_{X\sim\Bin(T,1/6)}[\mathcal{M}(X)]-\E_{X\sim\Bin(T/3,1/2)}[\mathcal{M}(X)]\right|<0.123.$$
Since we have shown that the mechanism outputs \acc with probability at least $99\%$ in $HY$, the mechanism outputs \acc with probability at least $86.7\%$ in $LX$ when the candidate-friendly agents deviate.
Since environment $LX$ happens with probability $0.49$, the expected utility for each candidate-friendly agent $t$ is at least $0.49\times 86.7\%\times v_t(L,\acc)=0.84966$.
However, without deviating, \acc will be output with probability at most $0.98\cdot1\%+0.02=0.0298$, and the expected utility for each candidate-friendly agent $t$ is upper-bounded by $0.0298\times v_t(H,\acc)=0.0894$.
We have seen that the candidate-friendly agents receive a utility gain of at least $0.76026>\tau$.
\end{proof}

As a remark, our impossibility result Theorem~\ref{thm:impossibility} holds even for randomized mechanism.
If mechanism can be randomized, Proposition~\ref{prop:anonymous} becomes that $\{\alpha_F,\alpha_U,\alpha_C\}\cup\{\beta_{F,\ell},\beta_{F,h},\beta_{U,\ell},\beta_{U,h},\beta_{C,\ell},\beta_{C,h}\}$ completely determines the \emph{probability} that the mechanism outputs \acc (or \rej).
In the proof of Theorem~\ref{thm:impossibility}, $\mathcal{M}(X)$ becomes the \emph{probability} that the mechanism outputs \acc, rather than either $0$ or $1$.
The remaining part of the proof is exactly the same.

\section{Remarks, Limitations and Future Work}
\label{sect:conclusion}
We presented a mechanism that elicits and aggregates the information and preferences of voters over two alternatives. In particular, voters' truthfully reporting their signals forms a strong Bayes Nash equilibrium, and in this case the mechanism outputs the alternative that is favored by the majority with overwhelming probability.   

We have assumed agents are Bayesian.
Although this assumption may not be completely realistic in practice, we believe the theoretical properties of our mechanism still hold to a certain extent.
For example, agents may not exactly predict $\alpha_CT_{hh}+\alpha_F$ or $\alpha_CT_{h\ell}+\alpha_F$ in practice, but it is reasonable to assume that their predictions are roughly around these two numbers, or in between.
If so, all the theoretical properties will still hold.
\citet{prelec2017solution} also assume Bayesian agents in the theoretical analysis of the surprisingly popular method, but their empirical experiments with human subjects suggest the method still works in practice.

We would also like to remark that, although our analysis assume $T$ is large, the failure probability in Theorem~\ref{thm:msw} and $\varepsilon$ in Theorem~\ref{thm:SBNE2} are exponentially small in $T$, making our mechanism applicable to the scenario with a small number of agents.

Our mechanism can be extended to the setting where a fixed fraction $\tau$ of acceptance votes is required to adopt a policy. For example, in many countries, constitutional amendments require a $2/3$ majority to pass.
To do this, we only need to change Step~4 of Mechanism~\ref{mechanism:M=N=2} such that $\bar\delta$ is the prediction with rank $\tau T$, and change Step~5 such that \acc is announced if more than $\tau$ fraction of agents report type $F$  and \rej is announced if more than $1 - \tau$ fraction of agents report type $U$. 

One limitation is that our mechanism only deals with two alternatives.
While this is natural in many scenarios (accept/reject, election with two candidates), extending our results to more than two alternatives is an interesting future direction, but it faces a multitude of hurdles: the median technique will not straightforwardly work, the ``surprisingly popular'' formalism faces impossibility results~\citep{prelec2017solution}, and  Gibbard-Satterthwaite social choice impossibility results apply.

Another interesting future direction is deployment. 
It would be interesting to test this mechanism in the real world, and then test to see if participants are, in aggregate, more happy when this mechanism is used as compared with a majority vote mechanism.
For example, groups could choose a movie to watch where different participants have different information about the potential movies.  In particular, not everyone has seen both movies.  Participants could be surveyed afterwards about how enjoyable the movie was.
In general, our Wisdom-of-the-Crowd-Voting mechanism could be tested in any place that currently uses majority voting to better aggregate information. 

Of course, it would suffer from some of the same drawbacks of majority voting: that the majority can impose their will on the minority.  It is not clear if either one of these enjoys fairness properties not included by the other, but that would be another direction of future inquiry.  

 

\bibliographystyle{plainnat}
\bibliography{reference}

\newpage
\appendix
\section{Comparison with Feddersen and Pesendorfer's Work}
\label{append:comparison}
\citet{Feddersen1994} consider a two-alternative setting similar to our model.
As mentioned before, \citet{Feddersen1994} consider the standard majority voting where each agent votes for an alternative, while assuming agents play a (Bayes) Nash equilibrium strategy profile.
We, on the other hand, design a more sophisticated mechanism to incentivize truth-telling.

Other than this difference, the state space, the signal space and the space of agents' types in \citet{Feddersen1994} are all continuous.
For Feddersen and Pesendorfer's continuous setting, in the Nash equilibrium, agents' strategies have three types: always vote for one alternative, always vote for the other, and vote the alternative based on the signal.
These are similar to our three types: candidate-friendly, contingent, candidate-unfriendly.
However, due to continuity, each agent needs to compute his/her type by solving an equation with a complicated Riemann integral (while agents' know their types directly according to their utility functions in our setting).
A phenomenon in their setting due to continuity is that the fraction of contingent voters in the Nash equilibrium approaches zero when the number of the voters goes to infinity.

Although agents can be classified by three types in both settings, we would like to clarify a fundamental difference in the motivation behind this classification.
In our setting, each agent's type reflects his/her preference over the two alternatives.
In Feddersen and Pesendorfer's setting, each agent ``chooses'' a type in a specific way so that the majority voting scheme outputs the correct alternative with high probability.
Therefore, in their setting, agents' types are chosen for collaboratively aggregating information, and should not be viewed as reflections of their preferences.
Although an agent's preference does affect his/her choice, the purpose for  choosing a type is for information aggregation, not for reflecting the preference.

At a high level, our mechanism includes some novel techniques, including the surprisingly popular technique and the median trick, to ensure the output of the correct alternative in the setting with strategic agents.
In Feddersen and Pesendorfer's setting, it may be surprisingly that the simple majority voting scheme is already enough for output the correct alternative. The reason behind this is that certain implicit techniques for guaranteeing the correct output are ``embedded'' into agents' strategic behaviors.
In other words, the agents are the ones who work out those techniques, not the mechanism.
That is why we mentioned in the introduction that the agents in Feddersen and Pesendorfer's setting need to have much more sophistication compared with the agents in our setting.

Another difference is that they are considering a Nash equilibrium strategy profile, while our mechanism satisfies the much stronger criterion that truth-telling strategies form a \emph{strong} Bayes Nash equilibrium.

\section{Extension to Non-binary Worlds and Signals}
\label{append:nonbinary}
In Sect.~\ref{sect:prelim}, we generalized the model to the setting with non-binary worlds and non-binary signals.
In Sect.~\ref{sect:nonbinaryWorlds}, we present our mechanism for the setting with binary signals and non-binary worlds.
In Sect.~\ref{sect:nonbinary-signals}, we show that the generalization to non-binary signals is straightforward.

\subsection{Model and Preliminaries}
\label{sect:prelim}
In our non-binary model, as in our binary model, $T$ agents are voting for two \emph{alternatives}, \acc and \rej (corresponding to ``accept'' and ``reject'').  However, in our non-binary model there is a set of $N$ possible \emph{worlds} (or \emph{states}) $\calW=\{1,\ldots,N\}$, where the higher the value the more \acc is preferred to  \rej.
Agents do not know which world is the actual world that they are in.
They have a prior common belief on the likelihood of each world.
Let $W$ be the actual world which is viewed as a random variable.
Let $P_n=\Pr(W=n)$ be the prior over worlds.
Each agent knows the values of $P_1,\ldots,P_n$ as prior beliefs.
We further assume $P_n>0$ for each $n$, for otherwise we can remove world $n$ from $\calW$ without loss of generality.

Each agent will then receives a \emph{signal} from the set $\calS=\{1,\ldots,M\}$.
Let $S_t$ be the random variable representing the signal that agent $t$ receives.
Given $W=n$, for any $n$, the signals agents receive have the same distribution and are conditionally independent.
Let $P_{mn}=\Pr(S_t=m\mid W=n)$ be the probability that signal $m$ will be received (by an arbitrary agent $t$) if the actual world is $n$.
The set of values $\{P_{mn}:m=1,\ldots,M; n=1,\ldots,N\}$ is known by all the agents.
Signals are positively correlated to the worlds:
\begin{equation}\label{eqn:postive_correlation}
    \Pr\left(S_t\geq m\mid W=n_1\right)=\sum_{m'=m}^MP_{m'n_1}>\Pr\left(S_t\geq m\mid W=n_2\right)=\sum_{m'=m}^MP_{m'n_2}
\end{equation}
for any worlds $n_1>n_2$, any signal $m$, and any agent $t$.

\gs{I added this.}
The remaining definitions for the non-binary model in this section are rather analogous to the binary case, but we include them for completeness.  

Each agent $t$ has a \emph{utility function} $v_t:\calW\times\{\acc,\rej\}\to\{0,1,\ldots,B\}$.
As mentioned earlier, a higher value of $W$ indicates \acc is more preferable: $v_t(n_1,\acc)>v_t(n_2,\acc)$ and $v_t(n_1,\rej)<v_t(n_2,\rej)$ for any $n_1,n_2\in\calW$ with $n_1>n_2$.
Since we can always rescale agents' utilities, for simplicity, we assume without loss of generality that agents' utilities are integers and bounded by $B\in\mathbb{Z}^+$.
Agents, with their prior beliefs and receiving signals, will have posterior beliefs about the distribution of $W$ and react to the mechanism in a way maximizing their expected utilities.

We assume $v_t(n,\acc)\neq v_t(n,\rej)$ for each agent $t$ and each $n\in\calW$, so that agents always strictly prefer one alternative over the other.  Given a world $n$, let $T(\acc, n)=\{t\mid v_t(\acc,n)>v_t(\rej,n)\}$ be the set of agents that prefer $\acc$ in world $n$ and let  $\alpha^{\acc}_n = \frac{|T(\acc, n)|}{|T|}$ be the fraction of agents that prefer $\acc$ in world $n$.  We can similarly define $T(\rej, n)$ and $\alpha^{\rej}_n = \frac{|T(\rej, n)|}{|T|} = 1 - \alpha^{\acc}_n$.  \gs{If we ever use these, put them in line.}
Naturally, $\alpha_n^\acc$ is increasing in $n$ (when the underlying quality of the candidate increases, more agents prefer \acc) and $\alpha_n^\rej$ is decreasing in $n$.
As before, we assume that the $\alpha^{\acc}_n$ and $\alpha^{\rej}_n$ are common knowledge, which is natural in many scenarios, including the faculty candidate hiring example.
If this assumption does not hold, results in Sect.~\ref{sect:unknowndistribution} show that we cannot achieve the truthful guarantee even under the setting $M=N=2$.


\gs{should we generically define $M(n)$ or only for the true world?} \gs{Is $\alpha_n$ okay terminology?}

For any world $n$, let 
$$M(n) =  \left\{\begin{array}{cc} \acc & \alpha^{\acc}_{n} > \frac{1}{2} \\  \rej & \mbox{otherwise} \end{array}\right.$$  
be the majority preference if the actual world were $n$.  We  assume that $T$ is an odd number to avoid ties.  
\bt{I added this definition}\gs{ok}
\begin{definition}
Given a utility profile $\{v_1,\ldots,v_T\}$ and letting $n^\ast$ be the actual world, we say $M(n^\ast)$ is \emph{the majority wish}.
\end{definition}
The goal is to output the majority wish $M(n^\ast)$, the alternative that is preferred by at least half of the agents in the actual world.

Our results will sometimes require $T$, the number of agents, to be sufficiently large, and it may be helpful to think of $T\rightarrow\infty$.  However, we will always assume that the parameters of the model: $B$, $\{P_n\}_{n\in \calW}$,  $\{P_{mn}\}_{m\in\calS, n\in \calW}$, and $\{\alpha_n^\acc,\alpha_n^\rej\}_{n \in \calW}$, do not depend on $T$ in any way.

In the faculty candidate hiring example, the worlds $\calW=\{1,\ldots,N\}$ describe the quality of the candidate, with $1$ being the worst and $N$ being the best.
The signals $S_t\in\{1,\ldots,M\}$ correspond to the impression of this candidate, with $S_t=1$ being the worst impression and $S_t=M$ being the best impression.
It is natural to assume that $S_t$'s are positively correlated to $W$, which agrees with our model.

\subsubsection{Candidate-Friendly, Contingent and Candidate-Unfriendly Agents}
Let $\calL=\{n\in\calW\mid \alpha_n^\acc<\frac12\}$ and $\calH=\{n\in\calW\mid\alpha_n^\acc>\frac12\}$.
Since we cannot have $\alpha_n^\acc=\frac12$ for an odd $T$, $\{\calL,\calH\}$ is a partition of $\calW$.
In addition, since $\alpha_n^\acc$ is increasing in $n$, there exists a threshold such that all those $n$ below the threshold belong to $\calL$ and all those $n$ above belong to $\calH$.
Indeed, $\calL$ is the set of ``low quality'' worlds where \rej is preferred, and $\calH$ is the set of ``high quality'' world is preferred.

For each agent $t$, define $\calL_t=\{n\in\calW\mid v_t(n,\rej)>v_t(n,\acc)\}$ and $\calH_t=\{n\in\calW\mid v_t(n,\acc)>v_t(n,\rej)\}$.
Then $\calL_t$ is the set of those low quality worlds based on agent $t$'s utility where $\rej$ is preferred, and $\calH_t$ is the set of those high quality worlds where $\acc$ is preferred for $t$.
Since $v_t(n,\acc)-v_t(n,\rej)$ is increasing in $n$ (the first term is increasing and the second term is decreasing), each agent $t$ also has a personal threshold that separate $\calW$ to $\calL_t,\calH_t$.
We can define the candidate-friendly agents, contingent agents and candidate-unfriendly agents based on whether the personal threshold is below, equal to, or above the average threshold.

We say an agent $t$ is candidate-friendly if $\calH_t\cap\calL\neq\emptyset$.
This says that there exists a world $n\in\calW$ where the fraction of agents preferring \acc is below $1/2$ (i.e., $\alpha_n^\acc<\frac12$) but $t$ still prefers \acc.
Equivalently, an agent $t$ is candidate-friendly if $\calL_t\subsetneq\calL$, or $\calH\subsetneq\calH_t$.
Correspondingly, an agent $t$ is candidate-unfriendly if $\calL_t\cap\calH\neq\emptyset$, or equivalently, $\calL\subsetneq\calL_t$, or $\calH_t\subsetneq\calH$.
An agent $t$ is contingent if $\calL_t=\calL$, or equivalently, $\calH_t=\calH$.
We still use $F$, $C$, $U$ to denote the three types of agents, and use $\alpha_F,\alpha_C,\alpha_U$ to denote their fractions as before.
As a remark, a candidate-friendly agent (candidate-unfriendly agent resp.) does not always prefer \acc (\rej resp.) like before, (s)he merely has a threshold below (above resp.) the average.

Let $L=\max\{n\in\calL\}$ be the maximum world where \rej is preferred by the majority, and $H=\min\{n\in\calH\}$ be the minimum world where \acc is preferred by the majority.
We clearly have $H=L+1$.
For each agent $t$, let $L_t=\max\{n\in\calL_t\}$ be the maximum world where \rej is preferred, and let $H_t=\min\{n\in\calH_t\}$ be the minimum world where \acc is preferred.
Set $L_t=0$ if $\calL_t=\emptyset$ and $H_t=N+1$ if $\calH_t=\emptyset$.
Clearly, $H_t=L_t+1$.

In the binary setting, in the proof of Theorem~\ref{thm:msw}, we have discussed three cases: 1) $\alpha_F>0.5$, 2) $\alpha_U>0.5$ and 3) $\alpha_F<0.5$ and $\alpha_U<0.5$.
However, in the non-binary setting here, by the way we define the three types of agents, we are always in the third case.
\begin{proposition}\label{prop:nonbinarymedian}
$\alpha_F<0.5$ and $\alpha_U<0.5$.
\end{proposition}
\begin{proof}
We will only show $\alpha_F<0.5$, as the proof for $\alpha_U<0.5$ is similar.
For a candidate-friendly agent $t$, we have $\calL_t\subsetneq\calL$, so $L_t<L$.
Therefore, $t$ would prefer $\acc$ if the actual world is $L$.
Conversely, if an agent would prefer $\acc$ in the world $L\in\calL$, (s)he must be a candidate-friendly agent by our definition.
Therefore, the set of candidate-friendly agents is exactly the set of agents who would prefer $\acc$ in world $L$.
Thus, $\alpha_F=\alpha_L^\acc$.
The proposition follows from $\alpha_L^\acc<\frac12$ implied by the definition of $\calL$.
\end{proof}

\subsubsection{Additional Notations}  
Given a strategy profile $\Sigma=\{\sigma_1,\ldots,\sigma_T\}$ and a mechanism $\mathcal{M}$, let $\lambda_n^\mathcal{\acc, M}(\Sigma)$ be the probability that alternative \acc is announced as the winner given the actual world is $n$, then $\lambda_n^\mathcal{\rej, M}(\Sigma)=1-\lambda_n^\mathcal{\acc, M}(\Sigma)$ is the probability that alternative \rej wins given the actual world is $n$.
We will omit the superscript $\mathcal{M}$ when it is clear what mechanism we are discussing.

All the agents' \emph{ex-ante} utilities depend exclusively on $\lambda_1^\acc(\Sigma),\ldots,\lambda_N^\acc(\Sigma)$ (or $\lambda_1^\rej(\Sigma),\ldots,\lambda_N^\rej(\Sigma)$), and each agent $t$'s utility is given by
\begin{equation}
u_t(\Sigma)=\sum_{n=1}^NP_n\left(\lambda_n^\acc(\Sigma)v_t(n,\acc)+\lambda_n^\rej(\Sigma)v_t(n,\rej)\right)\label{eqn:ut_1}
\end{equation}
By substituting $\lambda_n^\mathcal{\rej}(\Sigma)=1-\lambda_n^\mathcal{\acc}(\Sigma)$,
\begin{equation}
u_t(\Sigma)=\sum_{n=1}^NP_nv_t(n,\rej)+\sum_{n=1}^NP_n\lambda_n^\acc(\Sigma)(v_t(n,\acc)-v_t(n,\rej)).\label{eqn:ut_2}
\end{equation}

We will always use $\Sigma^\ast=\{\sigma_1^\ast,\ldots,\sigma_T^\ast\}$ to denote the truthful strategy profile.

\gs{add notation table.}

\gs{Also update above with new lambda.}

\bt{done}

Table~\ref{tab:notation} lists all the frequently used notations.

\begin{table}[ht]
    \centering
    \begin{tabular}{ll}
    \hline
    notation & meaning\\
    \hline
        $\calW=\{1,\ldots,N\}$ & the set of all worlds \\
        $\calS=\{1,\ldots,M\}$ & the set of all signals\\
        $P_n$ & the prior belief for the probability the actual world is $n$\\
        $P_{mn}$ & the probability of receiving signal $m$ under world $n$\\
        $v_t(n,\acc),v_t(n,\rej)$ & the (\emph{ex-post}) utility for agent $t$ for alternative \acc, \rej if the actual world is $n$\\
        $u_t(\Sigma)$ & the (\emph{ex-ante}) expected utility for agent $t$ given strategy profile $\Sigma$\\
        $\alpha_n^\acc,\alpha_n^\rej$ & the fraction of agents preferring \acc, \rej under world $n$\\
        $M(n)$ & the majority favored alternative under world $n$\\
        $\calL$ & the set of worlds where more than half of the agents prefer \rej\\
        $\calH$ & the set of worlds where more than half of the agents prefer \acc\\
        $\calL_t$ & the set of worlds where \rej is preferred for agent $t$\\
        $\calH_t$ & the set of worlds where \acc is preferred for agent $t$\\
        $L$ & the maximum world where \rej is preferred by the majority\\
        $H$ & the minimum world where \acc is preferred by the majority\\
        $L_t$ & the maximum world where \rej is  preferred for agent $t$\\
        $H_t$ & the minimum world where \acc is preferred for agent $t$\\
        $F,C,U$ & candidate-friendly agents, contingent agents, candidate-unfriendly agents\\
        $\alpha_F,\alpha_C,\alpha_U$ & fractions of the three types of agents\\
        $\lambda_n^\acc(\Sigma),\lambda_n^\rej(\Sigma)$ & the probability a given mechanism outputs \acc, \rej for strategy profile $\Sigma$\\
        $\Sigma^\ast$ & the truthful strategy profile\\
        \hline
    \end{tabular}
    \caption{Table of notations.}
    \label{tab:notation}
\end{table}

\subsection{Non-binary Worlds}
\label{sect:nonbinaryWorlds}
In this section, we consider the generalization to the setting with more than two worlds $N>2$, while keeping the binary signal assumption $M=2$.
We will see in the next section that the generalization to non-binary signals is simple.
For this section, we will use $\ell$ to denote signal $1$ and $h$ to denote signal $2$.

In the case $M=N=2$, we have asked each agent his/her received signal, type, and posterior belief on the fraction of agents who will report signal 1.
In the case $N>2$ here, while it is still natural to ask an agent for his/her signal, asking for a posterior prediction and keeping the mechanism as before will not work here.
In particular, this will make Theorem~\ref{thm:msw} fail.
To reason this intuitively, it is easy to see that, if the mechanism is required to output the alternative favored by the majority, the mechanism will output $\acc$ if the actual world is in $\calH$ and output $\rej$ if the actual world is in $\calL$.
To ensure this, we need to make sure the median of the posterior prediction is between $P_{hL}$ and $P_{hH}$.
However, while this is true for $N=2$ as Theorem~\ref{thm:keyInequality} suggests (in fact, all the possible posterior predictions, $T_{h\ell}$ and $T_{hh}$, are between $P_{hL}$ and $P_{hH}$), this is not necessarily true for $N>2$.

As a solution to this issue, for each contingent agent $t$, we ask him/her for a value between $P_{hL_t}$ and $P_{hH_t}$ (agents with $\calL_t=\emptyset$ report a value between $0$ and $P_{h1}$ and agents with  $\calH_t=\emptyset$ report a value between $P_{hN}$ and $1$), and the median of these values will be between $P_{hL}$ and $P_{hH}$.
A natural way to ask an agent for this value can be, \emph{please give a percentage value $q$ such that you would like alternative \acc if the fraction of agents reporting signal $h$ is more than $q$ percent, and you would like alternative \rej otherwise.}

Our mechanism is presented in Mechanism~\ref{mechanism:general}.\footnote{We can consider the same modification as in the binary case (Section~\ref{sect:alternative_mechanism}).}

\begin{algorithm}
\caption{The Wisdom-of-the-Crowd-Voting Mechanism for $N>2$ and $M=2$}
\label{mechanism:general}
\begin{algorithmic}[1]
    \STATE Each agent $t$ reports to the mechanism the signal (s)he receives (either $\ell$ or $h$), the type (either $F$, $C$, or $U$).
    \STATE If agent $t$ reports type $F$, his/her reported signal will be automatically treated as $h$; if agent $t$ reports type $U$, his reported signal will be automatically treated as $\ell$.
    \STATE If an agent reports type $C$, ask him/her to report a value $q_t\in[0,1]$ such that (s)he would like \acc \emph{if and only if} the fraction of agents who report $h$ is more than $q_t$.  \emph{The value $q_t$ should be given with the treatment in the previous step being considered, and the mechanism makes this clear to the agents.} For an agent with type $F$, set $q_t=0$. For an agent with type $U$, set $q_t=1$.
    \STATE Compute the \emph{median} of those $q_t$, denoted by $\bar{q}$.
    \STATE If the fraction of the agents reporting $h$ is more than the median $\bar{q}$, announce \acc being the winning alternative; otherwise, announce \rej being the winning alternative.
\end{algorithmic}
\end{algorithm}

In our faculty candidate hiring example, the questionnaire corresponding to the mechanism looks like the following:

\begin{enumerate}
    \item Choose one of the following: as compared with the average faculty member, independent of the candidate's qualification:
    \begin{enumerate}
        \item[I] I am more predisposed toward rejection;   
        \item[II] I am more predisposed toward acceptance;
        \item[III] Neither I nor II, i.e., I am with the average faculty member.
    \end{enumerate}
    \item What is your impression of this candidate during the individual interview between you and this candidate?
    \begin{enumerate}
        \item I had a good impression.
        \item I did not have a good impression.
    \end{enumerate}
   
    \item Your ``provisional ballot'' will be cast as follows: 
    \begin{itemize}
        \item Accept: if you choose II in Question 1, or, if you choose III in Question 1 and (a) in Question 2;
        \item Reject: if you choose I in Question 1, or, if you choose III in Question 1 and (b) in Question 2.
    \end{itemize}
    \item If you answered III for Question 1, what fraction of provisional ballots do you predict will be cast for accept?  
    %
\end{enumerate}

\medskip

Next, we will prove that the mechanism outputs the alternative favored by the majority with high probability, and the truthful strategy profile is an $\varepsilon$-strong Bayes Nash equilibrium for $\varepsilon=o(1)$.
For simplicity and clarity in describing the ideas behind the proofs, we will not perform the Chernoff bound analyses as in Section~\ref{sect:omitted_proofs_mainTheorem}, and we will assume $T\rightarrow\infty$ as it is in Sect.~\ref{sect:twoworldstwosignals}.
As a result, for any world $n$, the fractions of agents receiving signal $\ell$ and $h$ are, almost surely, $P_{\ell n}$ and $P_{hn}$ respectively.

\begin{theorem}\label{thm:maxsw_general}
Suppose $T\rightarrow\infty$. If all the agents play the truthful strategy, then our mechanism outputs an alternative favored by more than half of the agents.
\end{theorem}
\begin{proof}
Suppose all the agents report truthfully.
Let $n$ be the actual world.
We need to show that \acc is announced if and only if $n\in\calH$.
We assume $n\in\calH$ without loss of generality, as the analysis for $n\in\calL$ is similar.
When all the agents play the truthful strategy profile $\Sigma^\ast$, the fraction of agents reporting signal $h$ is $\alpha_{F}+\alpha_{C}\cdot P_{hn}$.

On the other hand, for any contingent agent $t$, (s)he prefers \acc if $n\in\calH_t=\calH$, and (s)he prefers \rej if $n\in\calL_t=\calL$.
If (s)he were asked to give a value such that (s)he would like \acc if and only if the fraction of agents who \emph{receive} signal $h$ is less than this value, (s)he would have report a value between $P_{hL}$ and $P_{hH}$.
Now, considering that, as instructed by the mechanism, those $\alpha_{F}\cdot T$ (resp. $\alpha_{U}\cdot T$) agents will always \emph{report} signal $h$ (resp. signal $\ell$) regardless of what they receive,
(s)he will report $q_t\in(\alpha_{F}+\alpha_C\cdot P_{hL},\alpha_{F}+\alpha_C\cdot P_{hH})$.
By our mechanism, for any candidate-friendly agent $t$, we have $q_t=0$, and for any candidate-unfriendly agent $t$, we have $q_t=1$.

Since $\alpha_F<0.5$ and $\alpha_U<0.5$ (Proposition~\ref{prop:nonbinarymedian}), it is then easy to see that the median $\bar{q}$ is between $\alpha_{F}+\alpha_{C}\cdot P_{hL}$ and $\alpha_{F}+\alpha_{C}\cdot P_{hH}$, which is less than the fraction of agents reporting signal $h$ (which is $\alpha_{F}+\alpha_{C}\cdot P_{hn}$ as computed earlier).
The last step of our mechanism will make sure \acc is output.
\end{proof}

As a remark, if we do not assume $T\rightarrow\infty$, to show that the statement in Theorem~\ref{thm:maxsw_general} fails with an exponentially low probability, we need to make an extra assumption that the median $\bar{q}$ is not exponentially close to the two endpoints $\alpha_{F}+\alpha_{C}\cdot P_{hL}$ and $\alpha_{F}+\alpha_{C}\cdot P_{hH}$.
This is a natural assumption, as an agent's reported value should not depend on $T$.
In addition, in practice, it is natural to expect that most agents will report values that are around the midpoint of the interval $(\alpha_{F}+\alpha_{C}\cdot P_{hL_t},\alpha_{F}+\alpha_{C}\cdot P_{hH_t})$.

Next, we show that our mechanism satisfies the truthful property.
Again, we consider $T\rightarrow\infty$.

\begin{theorem}\label{thm:truthful_generalN}
Suppose $T\rightarrow\infty$. The truthful strategy profile forms a strong Bayes Nash Equilibrium.
\end{theorem}

The ideas behind the proof of Theorem~\ref{thm:truthful_generalN} is similar as before.
Firstly, those contingent agents will not deviate from the truthful strategy, as their utilities have already been maximized.
Secondly, we can prove a lemma similar to Lemma~\ref{lem:conflict_limit} showing that there is a conflict of interest between an arbitrary candidate-friendly agent and an arbitrary candidate-unfriendly agent.
This shows that the set of deviating agents $D$ can only contain either candidate-friendly agents or candidate-unfriendly agents.
This further implies that more than half of the agents are truth-telling.
Finally, the use of median in our mechanism ensures that less than half of the agents' deviating cannot change the output alternative in their favored direction.

\begin{lemma}\label{lem:conflict_general_N}
Let $\Sigma^\ast$ be the truthful strategy profile and $\Sigma'$ be an arbitrary strategy profile. Let $t_1$ be an arbitrary candidate-friendly agent and $t_2$ be an arbitrary candidate-unfriendly agent.
Suppose $T\rightarrow\infty$.
We have
\begin{enumerate}
    \item If $u_{t_1}(\Sigma')>u_{t_1}(\Sigma^\ast)$, then $u_{t_2}(\Sigma')<u_{t_2}(\Sigma^\ast)$.
    \item If $u_{t_1}(\Sigma')<u_{t_1}(\Sigma^\ast)$, then $u_{t_2}(\Sigma')>u_{t_2}(\Sigma^\ast)$
\end{enumerate}
\end{lemma}
The proof of this lemma is more involved than that of Lemma~\ref{lem:conflict_limit}, but the ideas behind are similar.
\begin{proof}
We will only prove (1), as the proof for (2) is similar.
For ease of notation, in this proof, we assume $t_1=1$ and $t_2=2$ without loss of generality.
Suppose $u_1(\Sigma')>u_1(\Sigma^\ast)$, and we aim to show $u_2(\Sigma')<u_2(\Sigma^\ast)$.
Theorem~\ref{thm:maxsw_general} implies that $\lambda_n(\Sigma^\ast)=0$ for all $n\in\calL$ and $\lambda_n(\Sigma^\ast)=1$ for all $n\in\calH$.
Firstly, we show that
\begin{equation}\label{eqn:conflict_general_N_1}
    \sum_{n\in\calL\cap\calH_1}P_n\left(\lambda_n^\acc(\Sigma')-\lambda_n^\acc(\Sigma^\ast)\right)>\sum_{n\in\calH\cap\calL_2}P_n\left(\lambda_n^\acc(\Sigma^\ast)-\lambda_n^\acc(\Sigma')\right).
\end{equation}
This is because
\begin{align*}
    0&<u_1(\Sigma')-u_1(\Sigma^\ast)\tag{by our assumption}\\
    &=\sum_{n=1}^NP_n\left(\lambda_n^\acc(\Sigma')-\lambda_n^\acc(\Sigma^\ast)\right)(v_1(n,\acc)-v_1(n,\rej))\tag{by \eqref{eqn:ut_2}}\\
    &\leq\sum_{n\in\calW\setminus(\calL_1\cup\calH_2)}P_n\left(\lambda_n^\acc(\Sigma')-\lambda_n^\acc(\Sigma^\ast)\right)(v_1(n,\acc)-v_1(n,\rej))\tag{$\dag$}\\
    &=\sum_{n\in\calL\cap\calH_1}P_n\left(\lambda_n^\acc(\Sigma')-\lambda_n^\acc(\Sigma^\ast)\right)(v_1(n,\acc)-v_1(n,\rej))\\
    &\qquad-\sum_{n\in\calH\cap\calL_2}P_n\left(\lambda_n^\acc(\Sigma^\ast)-\lambda_n^\acc(\Sigma')\right)(v_1(n,\acc)-v_1(n,\rej))\\
    &\leq\sum_{n\in\calL\cap\calH_1}P_n\left(\lambda_n^\acc(\Sigma')-\lambda_n^\acc(\Sigma^\ast)\right)(v_1(L,\acc)-v_1(L,\rej))\\
    &\qquad-\sum_{n\in\calH\cap\calL_2}P_n\left(\lambda_n^\acc(\Sigma^\ast)-\lambda_n^\acc(\Sigma')\right)(v_1(L,\acc)-v_1(L,\rej))\tag{$\ddag$},
\end{align*}
which implies \eqref{eqn:conflict_general_N_1}, where both Step ($\dag$) and ($\ddag$) are based on the following facts.
In particular, ($\dag$) is based on the first two facts, and ($\ddag$) is based on the first and the third facts.
\begin{itemize}
    \item for $n\in\calL$, $\lambda_n^\acc(\Sigma')-\lambda_n^\acc(\Sigma^\ast)=\lambda_n^\acc(\Sigma')-0\geq0$; for $n\in\calH$, $\lambda_n^\acc(\Sigma^\ast)-\lambda_n^\acc(\Sigma')=1-\lambda_n^\acc(\Sigma')\geq0$.
    \item $v_1(n,\acc)-v_1(n,\rej)$ is negative for $n\in\calL_1$ and is positive for $n\in\calH_2$. Notice that this is also true for $v_2$.
    \item the expression $v_1(n,\acc)-v_1(n,\rej)$ is increasing in $n$. This is true for any agent $t$. In particular, for each agent $t$, $v_t(n,\acc)$ is increasing in $n$ and $v_t(n,\rej)$ is decreasing in $n$.
\end{itemize}
Next, we show that \eqref{eqn:conflict_general_N_1} implies $u_2(\Sigma')<u_2(\Sigma^\ast)$.
By the same calculations and analyses above, we have
\begin{align*}
    u_2(\Sigma')-u_2(\Sigma^\ast) 
    \leq&\sum_{n\in\calL\cap\calH_1}P_n\left(\lambda_n^\acc(\Sigma')-\lambda_n^\acc(\Sigma^\ast)\right)(v_2(L,\acc)-v_2(L,\rej))\\
    &-\sum_{n\in\calH\cap\calL_2}P_n\left(\lambda_n^\acc(\Sigma^\ast)-\lambda_n^\acc(\Sigma')\right)(v_2(L,\acc)-v_2(L,\rej))\tag{same calculations above}\\
    <&0,\tag{by  $v_2(L,\acc)-v_2(L,\rej)<0$ and \eqref{eqn:conflict_general_N_1}}
\end{align*}
which implies the lemma.
\end{proof}

Now we are ready to prove Theorem~\ref{thm:truthful_generalN}.
\begin{proof}[Proof of Theorem~\ref{thm:truthful_generalN}]
Suppose otherwise and there is a set of deviating agents $D$.
Let $\Sigma'$ be the profile after the deviation of agents in $D$.
Firstly, we show that $D$ cannot contain a contingent agent.
Notice that such an agent's utility has already been maximized by the truthful profile $\Sigma^\ast$.
Suppose $\lambda_n^\acc(\Sigma')\neq\lambda_n^\acc(\Sigma^\ast)$ for certain $n$.
It must be that $\lambda_n^\acc(\Sigma')>\lambda_n^\acc(\Sigma^\ast)=0$ if $n\in\calL$, and $\lambda_n^\acc(\Sigma')<\lambda_n^\acc(\Sigma^\ast)=1$ if $n\in\calH$.
It is then easy to see that this agent's utility will decrease, which contradicts to 2 of Definition~\ref{def:SBNE}.
Suppose $\lambda_n^\acc(\Sigma')=\lambda_n^\acc(\Sigma^\ast)$ for all $n$.
We have $u_t(\Sigma')=u_t(\Sigma^\ast)$ for every agent $t$.
This already contradicts to 3 of Definition~\ref{def:SBNE}.

Next, Lemma~\ref{lem:conflict_general_N} ensures that $D$ cannot contain both a candidate-friendly agent and a candidate-unfriendly agent.
Assume without loss of generality that $D$ only contains candidate-friendly agents.
In order to maximize the chance that \acc is output, those candidate-friendly agents need to maximize the fraction of agents reporting signal $h$ and minimize the median $\bar{q}$.
However, the mechanism always does this for them in the truth-telling profile $\Sigma^\ast$: Step~2 makes sure they report signal $h$, and Step~3 makes sure they report $q_t=0$.
Therefore, those candidate-friendly agents' utilities are maximized by truth-telling, which contradicts to our assumption for $D$.
\end{proof}

\subsection{Non-binary Signals} \label{sect:nonbinary-signals}
There is a simple reduction from the non-binary signal setting to the binary-signal setting.
Suppose the signal space is $\{1,\ldots,M\}$.
To reduce it to a binary signal space $\{\ell,h\}$, we set an arbitrary non-integer number $s_\top$ between $1$ and $M$.
All the signals less than $s_\top$ are reduced to the ``bad'' signal $\ell$, and all the signals greater than $s_\top$ are reduced to the ``good'' signal $h$.
The mechanisms in the previous sections can be adapted to the setting here.
The mechanisms are the same as before, except for the following change: whenever the mechanism asks an agent for a binary signal in the previous setting, the mechanism asks the agent whether the signal (s)he received is less than or more than $s_\top$, which corresponds to signal $\ell$ and $h$ respectively.

For the mechanism in Section~\ref{sect:nonbinaryWorlds}, all the properties, including that the mechanism outputs the alternative favored by more than half of the agents and that the truth-telling strategy profile form a strong Bayes Nash Equilibrium, continue to hold in the non-binary signal setting with exactly the same proofs.

If we are dealing with binary world $\calW=\{L,H\}$, for the mechanism in Section~\ref{sect:twoworldstwosignals}, it is easy to see that these properties also continue to hold here if we prove the following inequality that is similar to the one in Theorem~\ref{thm:keyInequality}:
\begin{equation}\label{eqn:key_nonbinary}
    \forall m\in\{1,\ldots,M\}:P_{h H}>T_{h m}>P_{h L},
\end{equation}
where $P_{hL}$ and $P_{hH}$ are the probabilities that a signal above $s_\top$ is received if the actual world is $L$ and $H$ respectively, and $T_{hm}$ is the probability that an agent who receives signal $m$ believes that another agent will receive a signal that is more than $s_\top$.
Intuitively, if (\ref{eqn:key_nonbinary}) holds, all the agents' posterior predictions are still between $P_{hL}$ and $P_{hH}$, and the majority wish will still be ``surprisingly popular''.
The proof of (\ref{eqn:key_nonbinary}) is by straightforward Bayesian analysis, and is left to the readers.

\end{document}